%% file: main.tex
\documentclass[a4paper,runningheads]{llncs}

\newif\iflongversion\longversiontrue

\usepackage{amssymb}
\usepackage{stmaryrd}
\usepackage{mathtools}
\usepackage{thmtools}
\usepackage{thm-restate}
\usepackage{proof}
\usepackage[dvipsnames]{xcolor}
\usepackage{xparse}
\usepackage[hidelinks]{hyperref}
\usepackage{cleveref}
\usepackage[basic]{complexity}
\usepackage{csquotes}
\usepackage{subcaption}
\usepackage{wasysym}
\usepackage{xspace}
\usepackage{comment}
\usepackage[inline]{enumitem}
\usepackage{cite}

\iflongversion
\usepackage[final]{crossreftools}
\fi

\newclass{\THREEEXP}{3EXP}
\newclass{\TWOEXP}{2EXP}

\iflongversion
\allowdisplaybreaks
\else
\usepackage[mathlines]{lineno}
\newcommand*\patchAmsMathEnvironmentForLineno[1]{%
  \expandafter\let\csname old#1\expandafter\endcsname\csname #1\endcsname
  \expandafter\let\csname oldend#1\expandafter\endcsname\csname end#1\endcsname
  \renewenvironment{#1}%
  {\linenomath\csname old#1\endcsname}%
  {\csname oldend#1\endcsname\endlinenomath}}%
\newcommand*\patchBothAmsMathEnvironmentsForLineno[1]{%
  \patchAmsMathEnvironmentForLineno{#1}%
  \patchAmsMathEnvironmentForLineno{#1*}}%
\AtBeginDocument{%
  \patchBothAmsMathEnvironmentsForLineno{equation}%
  \patchBothAmsMathEnvironmentsForLineno{align}%
  \patchBothAmsMathEnvironmentsForLineno{flalign}%
  \patchBothAmsMathEnvironmentsForLineno{alignat}%
  \patchBothAmsMathEnvironmentsForLineno{gather}%
  \patchBothAmsMathEnvironmentsForLineno{multline}%
}
\linenumbers
\fi

\RequirePackage{tikz}
\usetikzlibrary{
    automata,
    arrows.meta,
    calc,
    positioning,
    backgrounds,
    shapes.geometric,
    shapes.misc,
    plotmarks,
    patterns,
    decorations.pathreplacing,
}
\tikzset {
  edge with arrow/.style = {
    ->,
    >=stealth,
    shorten >=1pt,
  },
  directed/.style = {
    edge with arrow,
    node distance=2cm,
    on grid,
    semithick,
    double distance=1.5pt,
  },
  automaton/.style = {
    directed,
    auto,	
    initial text={},
    state/.append style = {
        ellipse,
    },
  },
  timeline/.style = {
    very thick,
  },
  hole/.style = {
    timeline,
    draw = none,
  },
  block/.style = {
    draw,
  },
  blockBeforeKilled/.style = {
    draw,
    dashed,
  },
  nullTimerKilled/.style = {
    circle,
    fill,
    inner sep = 2pt,
  },
  nonNullTimerKilled/.style = {
    draw,
    cross out,
    solid,
    inner sep = 2pt,
  },
  delay/.style = {
    {Stealth[length=4pt]}-{Stealth[length=4pt]},
    shorten > = 0pt,
  },
  blockGraph/.style = {
    directed,
  }
}

\newcommand{\iffMSO}{\leftrightarrow}
\newcommand{\impliesMSO}{\rightarrow}

\newcommand{\timeout}[1]{\mathit{to}[#1]}
\newcommand{\killSymbol}[1]{\mathit{di}[#1]}
\newcommand{\toevents}[1]{\mathit{TO}[#1]}
\newcommand{\subsets}[1]{{\mathcal{P}}(#1)} 
\newcommand{\nat}{{\mathbb N}}
\newcommand{\natplus}{\nat^{>0}}
\newcommand{\nnr}{{\mathbb R}^{\geq 0}}

\newcommand{\automaton}[1]{\mathcal{#1}}
\newcommand{\M}{\automaton{M}}
\newcommand{\A}{\automaton{A}}
\newcommand{\N}{\automaton{N}}

\newcommand{\AwT}{AT\xspace}
\newcommand{\AwTs}{ATs\xspace}
\newcommand{\robust}{race-avoiding\xspace}
\newcommand{\Robust}{Race-avoiding\xspace}
\newcommand{\valuation}{\kappa}
\newcommand{\Val}[1]{\mathsf{Val}({#1})} 
\newcommand{\dom}[1]{{\textsf{dom}}(#1)} 

\newcommand{\truns}[1]{\mathit{truns}(#1)} 
\newcommand{\mtruns}[1]{\mathit{ptruns}(#1)} 
\newcommand{\untime}[1]{\mathit{untime}(#1)} 

\newcommand{\sink}{\mathit{sink}}
\newcommand{\done}{\mathit{done}}
\newcommand{\state}{q}
\newcommand{\tsymbol}{\mathit{symbol}}
\renewcommand{\clock}{\mathit{clock}}
\newcommand{\phase}{\mathit{phase}}

\newcommand{\go}{\mathit{go}}

\NewDocumentCommand{\clockConstraints}{O{C}}{\Phi(#1)}

\NewDocumentCommand{\values}{O{x}}{\mathrm{Values}(#1)} 

\newcommand{\regionRelation}{\cong}
\newcommand{\regionClass}[1]{\llbracket {#1} \rrbracket_{\regionRelation}}
\newcommand{\regionAutomaton}{\mathcal{R}}

\newcommand{\delaySymbol}{\tau}

\newcommand{\timerFate}{\gamma}
\newcommand{\nilTimerKilled}{\textup{\newmoon}}
\newcommand{\nonNilTimerKilled}{\times}

\newcommand{\elapsed}{relative elapsed time\xspace}
\newcommand{\distance}{\mathrm{reltime}}
\newcommand{\ext}[1]{\mathrm{ext}(#1)}

\NewDocumentCommand{\blockGraph}{O{}}{G_{#1}}

\newcommand{\blockMSO}{\mathrm{Block}}
\newcommand{\raceMSO}{\mathrm{Race}}
\newcommand{\firstMSO}{\mathrm{First}}
\newcommand{\lastMSO}{\mathrm{Last}}
\newcommand{\nextMSO}{\mathrm{Next}}
\newcommand{\nowigglingMSO}{\Phi}
\newcommand{\partition}{\mathrm{Partition}}

\newcommand{\maximal}{padded\xspace}

\def\orcidID#1{\smash{\href{http://orcid.org/#1}{\protect\raisebox{-1.25pt}{\protect\includegraphics{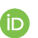}}}}}

\begin{document}

\title{Automata with Timers%
    \thanks{This work was supported by the Belgian FWO \enquote{SAILor} project (G030020N). Ga\"etan Staquet is a research fellow (Aspirant) of the Belgian F.R.S.-FNRS\@.
    The research of Frits Vaandrager was supported by NWO TOP project 612.001.852 ``Grey-box learning of Interfaces for Refactoring Legacy Software (GIRLS)''.	
    }%
}
\author{V\'eronique Bruy\`ere\inst{1}\orcidID{0000-0002-9680-9140}%
  \and Guillermo A. P\'erez\inst{2}\orcidID{0000-0002-1200-4952}%
  \and Ga\"etan Staquet\inst{1,2}\orcidID{0000-0001-5795-3265}%
  \and Frits W. Vaandrager\inst{3}\orcidID{0000-0003-3955-1910}%
}
\authorrunning{V. Bruy\`ere et al.}

\institute{University of Mons, Belgium\\
  \email{\{veronique.bruyere,gaetan.staquet\}@umons.ac.be}
  \and University of Antwerp -- Flanders Make, Belgium\\
  \email{guillermo.perez@uantwerpen.be}%
  \and Radboud University, The Netherlands\\
  \email{f.vaandrager@cs.ru.nl}
}

\maketitle

\begin{abstract}
  In this work, we study properties of deterministic finite-state automata with timers, a subclass of timed automata proposed by Vaandrager et al.\ as a candidate for an efficiently learnable timed model. We first study the  complexity of the configuration reachability problem for such automata and establish that it is \PSPACE-complete.
Then, as simultaneous timeouts (we call these, races) can occur in timed runs of such automata, we study the problem of determining whether it is possible to modify the delays between the actions in a run, in a way to avoid such races. 
The absence of races is important for modelling purposes and to 
streamline learning of automata with timers. We provide an effective characterization of when an automaton is \robust and establish that the related decision problem is in \THREEEXP{} and \PSPACE-hard.

  \keywords{Timed systems \and model checking \and reachability}
\end{abstract}

\input{intro}

\section{Preliminaries}\label{sec:def}

An \emph{automaton with timers} uses a finite set $X$ of \emph{timers}. Intuitively, a timer can be \emph{started} to any integer value to become \emph{active}.  Subsequently, its value is decremented as time elapses (i.e., at the same fixed rate for all timers). When the value of a timer reaches $0$, it \emph{times out} and it is no longer active.
Active timers can also be \emph{stopped}, rendering them inactive, too. 
Such an automaton, along any transition, can stop a number of timers and update a \emph{single} timer.  

Some definitions are in order.
We write $\toevents{X}$ to denote the set $\{\timeout{x} \mid x \in X\}$ of \emph{timeouts of} $X$. We denote by $I$ a
finite set of \emph{inputs}. We write $\hat{I}$ to denote the set $I \cup
\toevents{X}$ of \emph{actions}: either reading an input (an
\emph{input-action}), or processing a timeout (a \emph{timeout-action}).
Finally, we denote by $U = (X \times \natplus) \cup \{\bot\}$ the set of
\emph{updates}, where $(x,c)$ means that timer $x$ is started with value $c$, and $\bot$ stands for no timer update.

\begin{definition}[Automaton with timers]\label{def:AwT}
  An automaton with timers (\AwT, for short) is a tuple $\A =
  (X,I,Q,q_0,\chi,\delta)$ where:
  \begin{itemize}
    \item $X$ is a finite set of timers, $I$ a finite set of inputs,
    \item $Q$ is a finite set of states, with $q_0 \in Q$ the initial state,
    \item $\chi \colon Q \to \subsets{X}$, with $\chi(q_0) = \emptyset$, is a total function that assigns a finite set of active timers to each state,
    \item $\delta \colon Q \times \hat{I} \to Q \times U$ is a partial transition function that assigns a state and an update to each state-action pair, such that 
    \begin{itemize}
        \item  $\delta(q,i)$ is defined iff either $i \in I$ or there is a timer $x \in \chi(q)$ with $i = \timeout{x}$,
        \item if $\delta(q,i) = (q',u)$ with $i = \timeout{x}$ and $u = (y,c)$, then $x = y$ (when processing a timeout $\timeout{x}$, only the timer $x$ can be restarted). 
    \end{itemize}
  \end{itemize}
  Moreover, any transition $t$ of the form $\delta(q, i) = (q', u)$ must be such that
\begin{itemize}
  \item if $u = \bot$, then $\chi(q') \subseteq \chi(q)$ (all timers active in $q'$ were already active in $q$ in case of no timer update); moreover, if $i = \timeout{x}$, then $x \not\in\chi(q')$ (when the timer $x$ times out and is not restarted, then $x$ becomes inactive in $q'$);
  \item if $u = (x,c)$, then $x \in \chi(q')$ and $\chi(q')\setminus \{x\} \subseteq \chi(q)$  ((re)starting the timer \(x\) makes it active in \(q'\)).
  \end{itemize}
When a timer $x$ is active in $q$ and $i \neq \timeout{x}$, we say that the transition $t$ \emph{stops} $x$ if $x$ is inactive in $q'$, and that $t$ \emph{discards} $x$ if $t$ stops $x$ or restarts $x$. 
We write $q \xrightarrow[u]{i} q'$ if $\delta(q,i) = (q',u)$.
\end{definition}

\begin{example}\label{ex:AwT}
  \begin{figure}[t]
    \centering
    \input{figures/definitions/MMT.tex}
    \caption{An automaton with two timers $x_1,x_2$, such that $\chi(q_0) = \emptyset$, $\chi(q_1) = \{x_1\}$, $\chi(q_2) = \{x_1, x_2\}$, and \(\chi(q_3) = \{x_2\}\).}%
    \label{fig:AwT}
  \end{figure}
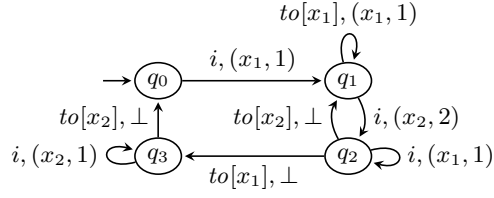

  An \AwT \(\A\) is shown in \Cref{fig:AwT} with set $X = \{x_1, x_2\}$ of timers and with set $I = \{i\}$ of inputs. In the initial state \(q_0\), no timer is active, while \(x_1\) is active in \(q_1\) and \(q_2\), and \(x_2\) is active in \(q_2\) and \(q_3\). That is,
  \(\chi(q_0) = \emptyset, \chi(q_1) = \{x_1\}, \chi(q_2) = \{x_1, x_2\}\), and \(\chi(q_3) = \{x_2\}\). Timer updates are shown in the transitions. For instance, \(x_1\) is started with value $1$ when going from \(q_0\) to \(q_1\).
  The transition looping on $q_2$ discards $x_1$ and restarts it with value $1$.
  \end{example}

\subsection{Timed semantics}

The semantics of an \AwT $\A$ is defined via an infinite-state labeled transition system that describes all possible configurations and transitions between them.

A \emph{valuation} is a partial function $\valuation \colon X \to \nnr$ that assigns nonnegative real numbers to timers. For $Y \subseteq X$, we write $\Val{Y}$ for the set of all valuations $\valuation$ such that $\dom{\valuation} = Y$.\footnote{Notation $\dom{f}$ means the domain of the partial function $f$.}  A \emph{configuration} of $\A$ is a pair $(q, \valuation)$ where $q \in Q$ and $\valuation \in \Val{\chi(q)}$. The \emph{initial configuration} is the pair $(q_0, \valuation_0)$ where $\valuation_0$ is the empty valuation since $\chi(q_0) = \emptyset$. If $\valuation \in \Val{Y}$ is a valuation in which all timers from $Y$ have a value of at least $d \in \nnr$, then $d$ units of time may elapse. We write $\valuation - d \in\Val{Y}$ for the valuation that satisfies $(\valuation -d)(x) = \valuation(x) -d$, for all $x \in Y$. The following rules specify the transitions between configurations $(q, \valuation), (q', \valuation')$.
\begin{gather}
  \label{Rule1}
  \infer{(q, \valuation) \xrightarrow{d} (q, \valuation - d)}
  {\forall x \colon \valuation(x) \geq d}
  \\
  \label{Rule2}
  \infer{(q, \valuation) \xrightarrow[\bot]{i} (q', \valuation')}
  {q \xrightarrow[\bot]{i} q', & i = \timeout{x} \Rightarrow \valuation(x) =
  0, & \forall y \in \chi(q') \colon \valuation'(y) = \valuation(y)}
  \\
  \label{Rule3}
  \infer{(q, \valuation) \xrightarrow[(x,c)]{i} (q', \valuation')}
  {q \xrightarrow[(x,c)]{i} q', & i = \timeout{x} \Rightarrow \valuation(x)
    = 0, & \forall y \in \chi(q') \colon \valuation'(y) = \begin{cases}
    c & \mbox{if } y = x\\
    \valuation(y) & \mbox{otherwise}
    \end{cases}
  }
\end{gather}
Transitions of type~\eqref{Rule1} are called \emph{delay transitions} (delay zero is allowed); those of type~\eqref{Rule2} and~\eqref{Rule3} are called \emph{discrete transitions} (\emph{timeout transitions} when $i = \timeout{x}$ and \emph{input transitions} otherwise). A \emph{timed run} of $\A$ is a sequence of configurations such that delay and discrete transitions alternate. The set $\truns{\A}$ of timed runs is defined inductively as follows.
\begin{itemize}
\item The sequence $(q_0,\valuation_0) \xrightarrow{d} (q_0, \valuation_0 - d)$ is in $\truns{\A}$.
\item Suppose $\rho (q, \valuation)$ is a timed run ending with configuration $(q, \valuation)$, then $\rho' = \rho (q, \valuation) \xrightarrow[u]{i} (q', \valuation') \xrightarrow{d} (q', \valuation' - d)$ is in $\truns{\A}$.
\end{itemize}
A timed run is also written as $\rho = (q_0,\kappa_0) ~ d_1 ~ i_1/u_1
~\dots ~ d_{n} ~ i_n/u_n ~ d_{n+1} ~(q,\kappa)$
such that only the initial configuration $(q_0,\kappa_0)$ and the last configurations $(q,\kappa)$ of $\rho$ are given. 
The \emph{untimed trace} of a timed run $\rho$, denoted $\untime{\rho}$, is the alternating sequence of states and actions from $\rho$, that is,
\(
  \untime{\rho} = q_0 ~ i_1 ~ \dots ~ i_n ~ q
\)
(we omit the valuations, the delays, and the updates).

\begin{example}\label{ex:trun}
  A sample timed run $\rho$ of the \AwT of Example~\ref{ex:AwT} is given below. Notice the transition with delay zero, indicating that two actions occur at the same time.
  \begin{align*}
    &\rho =
      (q_0, \emptyset)
      \! \xrightarrow{1} \!                     (q_0, \emptyset)
      \! \xrightarrow[(x_1,1)]{i} \!            (q_1, x_1 = 1)
      \! \xrightarrow{1} \!                     (q_1, x_1 = 0)
      \! \xrightarrow[(x_2, 2)]{i} \!           (q_2, x_1 = 0, x_2 = 2)\\
      &\! \xrightarrow{0} \!                    (q_2, x_1 = 0, x_2 = 2)
      \! \xrightarrow[\bot]{\timeout{x_1}}\!    (q_3, x_2 = 2)
      \! \xrightarrow{2} \!                     (q_3, x_2 = 0)
      \! \xrightarrow[\bot]{\timeout{x_2}} \!   (q_0, \emptyset)
      \! \xrightarrow{0.5} \!                   (q_0, \emptyset).
  \end{align*}
  The untimed trace of $\rho$ is $\untime{\rho} = q_0 ~ i ~ q_1 ~ i ~ q_2 ~\timeout{x_1} ~ q_3 ~ \timeout{x_2} ~q_0$.
\end{example}

\subsection{Blocks and races}\label{subsec:blocks}

In this section, given an \AwT $\A$, we focus on its timed runs $\rho = (q_0,\kappa_0) ~ d_1 ~ i_1/u_1 ~ \cdots$ $d_n ~ i_n/u_n ~ d_{n+1} ~ (q,\kappa)$ such that their first and last delays are non-zero and no timer times out in their last configuration, i.e., $d_1 > 0, d_{n+1} > 0$ and $\kappa(x) \ne 0$ for all $x \in \chi(q)$.\footnote{The reason for this choice will be clarified at the end of this section.} Such runs are called \emph{\maximal}, and we denote by $\mtruns{\A}$ the set of all \maximal timed runs of $\A$. To have a good intuition about \maximal timed runs, their decomposition into \emph{blocks} is helpful and will be often used in the proofs. A block is
composed of an input $i$ that starts a timer $x$ and of the succession of
timeouts and restarts of $x$, that $i$ induces inside a timed run. Let us
formalize this notion. Consider a \maximal timed run $\rho = (q_0,\kappa_0) ~ d_1 ~ i_1/u_1 ~ \dots ~ d_n ~ i_n/u_n ~ d_{n+1} ~ (q,\kappa)$ of an \AwT. Let $k,k'$
be such that $1 \leq k < k' \leq n$. We say that $i_k$ \emph{triggers}
$i_{k'}$ if there is a timer $x$ such that:
\begin{itemize}
\item $i_k$ (re)starts $x$, that is, $u_k = (x,c)$,
\item $i_{k'}$ is the action $\timeout{x}$, and 
\item there is no $\ell$ with $k < \ell < k'$ such that $i_\ell = \timeout{x}$
or $i_\ell$ discards $x$.
\end{itemize}
Note that $i_{k'}$ may restart $x$ or not, and if it does, $x$ later times out or is discarded. 

\begin{definition}[Block]
Let $\rho = (q_0,\kappa_0) ~ d_1 ~ i_1/u_1 ~ \dots ~ d_n ~ i_n/u_n ~
d_{n+1} ~ (q,\kappa)$ be a \maximal timed run of an \AwT.  A \emph{block} of $\rho$ is a pair
$B = ({k_1}  {k_2}  \dotsc  {k_m}, \timerFate)$ such that $i_{k_1},
i_{k_2},  \dotsc,  i_{k_m}$ is a maximal subsequence of actions of \(\rho\) such
that \(i_{k_1} \in I\), \(i_{k_\ell}\) triggers \(i_{k_{\ell+1}}\) for all
\(1 \leq \ell < m\), and \(\gamma\) is the \emph{timer fate} of \(B\)
defined as:  
\[
  \timerFate = \begin{cases}
    \bot & \text{if \(i_{k_m}\) does not restart $x$,} \\
    \nilTimerKilled & \text{if \(i_{k_m}\) restarts $x$ which is discarded by some $i_\ell$, with $k_m < \ell \leq n$},
    \\ & \text{when its value is zero,}\\
    \nonNilTimerKilled & \text{otherwise.}
  \end{cases}
\]
\end{definition}

In the timer fate definition, consider the case where \(i_{k_m}\) restarts $x$. For the purposes of \Cref{sec:wiggle_run}, it is convenient to know whether $x$ is later discarded or not, and in case it is discarded, whether this occurs when its value is zero ($\timerFate = \nilTimerKilled$). Hence, $\timerFate = \nonNilTimerKilled$ covers both situations where $x$ is discarded with non-zero value, and $x$ is still active in the last configuration $(q,\valuation)$ of $\rho$. Notice that in the latter case, $x$  has also non-zero value in $(q,\kappa)$ as $\rho$ is \maximal.
When no confusion is possible, we denote a block by a sequence of inputs rather
than the corresponding sequence of indices,
that is, \(B = (i_{k_1}  i_{k_2}  \dotsc  i_{k_m}, \timerFate)\).
In the sequel, we use notation $i \in B$ to denote an action $i$ belonging to the sequence of $B$.

By definition of an \AwT, recall that the same timer $x$ is restarted along a block $B$. Hence we also say that $B$ is an \emph{$x$-block}. Note also that the sequence of a block can be composed of a single input \(i \in I\).

As this notion of blocks is not trivial but plays a great role in this paper, let us give several examples illustrating multiple situations.

\begin{example}\label{ex:blocks}
Consider the timed run $\rho$ of Example~\ref{ex:trun} from the \AwT $\A$ depicted in Figure~\ref{fig:AwT}. It has two blocks: an $x_1$-block \(B_1 = (i ~ \timeout{x_1}, \bot)\) and an $x_2$-block \(B_2 = (i ~ \timeout{x_2}, \bot)\), both represented in \Cref{fig:ex:runs:rho}.\footnote{When using the action indices in the blocks, we have \(B_1 = (1 ~3, \bot)\) and \(B_2 = (2 ~ 4, \bot)\).} In this visual representation of the blocks, time flows left to right and is represented by the thick horizontal line.
A \enquote{gap} in that line indicates that the time is stopped, i.e., the delay between two consecutive actions is zero. We draw a vertical line for each input, and join together inputs belonging to a block by a horizontal (non-thick) line.

Consider another timed run $\sigma$ from $\A$:
  \begin{align*}
    \sigma = (q_0, \emptyset)
    &\! \xrightarrow{1} \!                      (q_0, \emptyset)
    \! \xrightarrow[(x_1, 1)]{i} \!             (q_1, x_1 = 1)
    \! \xrightarrow{1} \!                       (q_1, x_1 = 0)
    \! \xrightarrow[(x_1, 1)]{\timeout{x_1}} \! (q_1, x_1 = 1)\\
    &\! \xrightarrow{0} \!                      (q_1, x_1 = 1)
    \! \xrightarrow[(x_2, 2)]{i} \!             (q_2, x_1 = 1, x_2 = 2)
    \! \xrightarrow{1} \!                       (q_2, x_1 = 0, x_2 = 1)\\
    &\! \xrightarrow[\bot]{\timeout{x_1}} \!    (q_3, x_2 = 1)
    \! \xrightarrow{1} \!                       (q_3, x_2 = 0)
    \! \xrightarrow[\bot]{\timeout{x_2}} \!     (q_0, \emptyset)
    \! \xrightarrow{0.5} \!                     (q_0, \emptyset).
  \end{align*}
  This timed run has also two blocks represented in \Cref{fig:ex:runs:sigma}, such that \(B_1 = (i ~
  \timeout{x_1}~ \timeout{x_1}, \bot)\) with \(x_1\) timing out twice.  

\begin{figure}[t]
  \centering
  \begin{subfigure}{.23\textwidth}
    \centering
    \input{figures/blocks/rho.tex}
    \caption{Timed run \(\rho\).}%
    \label{fig:ex:runs:rho}
  \end{subfigure}
  \begin{subfigure}{.23\textwidth}
    \centering
    \input{figures/blocks/sigma.tex}
    \caption{Timed run \(\sigma\).}%
    \label{fig:ex:runs:sigma}
  \end{subfigure}
    \begin{subfigure}{.23\textwidth}
    \centering
    \input{figures/blocks/pi.tex}
    \caption{Timed run \(\pi\).}%
    \label{fig:ex:runs:pi}
  \end{subfigure}
  \begin{subfigure}{.23\textwidth}
    \centering
    \input{figures/blocks/tau.tex}
    \caption{Timed run \(\tau\).}%
    \label{fig:ex:runs:tau}
  \end{subfigure}
  \caption{Block representations of four timed runs.}%
  \label{fig:ex:runs}
\end{figure}
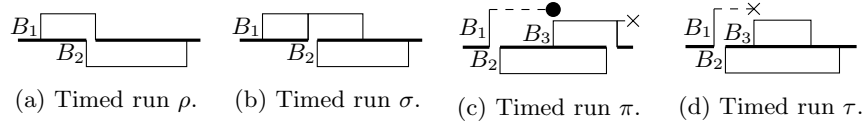

  We conclude this example with two other timed runs, $\pi$ and $\tau$, such that some of their blocks have a timer fate $\timerFate \neq \bot$. Let \(\pi\) and \(\tau\) be the timed runs:
  \begin{gather*}
    \begin{split}
        \pi ={}
        &(q_0, \emptyset)
        \! \xrightarrow{1} \!                       (q_0, \emptyset)
        \! \xrightarrow[(x_1, 1)]{i} \!             (q_1, x_1 = 1)
        \! \xrightarrow{0} \!                       (q_1, x_1 = 1)
        \! \xrightarrow[(x_2, 2)]{i}\!              (q_2, x_1 = 1, x_2 = 2)\\
        &\! \xrightarrow{1} \!                      (q_2, x_1 = 0, x_2 = 1)
        \! \xrightarrow[(x_1, 1)]{i} \!             (q_2, x_1 = 1, x_2 = 1)
        \! \xrightarrow{1} \!                       (q_2, x_1 = 0, x_2 = 0)\\
        &\! \xrightarrow[\bot]{\timeout{x_2}} \!    (q_1, x_1 = 0)
        \! \xrightarrow{0} \!                       (q_1, x_1 = 0)
        \! \xrightarrow[(x_1, 1)]{\timeout{x_1}} \! (q_1, x_1 = 1)
        \! \xrightarrow{0.5} \!                     (q_1, x_1 = 0.5)
    \end{split}\\
    \begin{split}
        \tau ={}
        &(q_0, \emptyset)
        \! \xrightarrow{1} \!                       (q_0, \emptyset)
        \! \xrightarrow[(x_1, 1)]{i} \!             (q_1, x_1 = 1)
        \! \xrightarrow{0} \!                       (q_1, x_1 = 1)
        \! \xrightarrow[(x_2, 2)]{i} \!             (q_2, x_1 = 1, x_2 = 2)\\
        &\! \xrightarrow{0.5} \!                    (q_2, x_1\! =\! 0.5, x_2\! =\! 1.5)
        \! \xrightarrow[(x_1, 1)]{i} \!             (q_2, x_1\! =\! 1, x_2\! =\! 1.5)
        \! \xrightarrow{1} \!                       (q_2, x_1\! =\! 0, x_2\! =\! 0.5)\\
        &\! \xrightarrow[\bot]{\timeout{x_1}} \!    (q_3, x_2 = 0.5)
        \! \xrightarrow{0.5} \!                     (q_3, x_2 = 0)
        \! \xrightarrow[\bot]{\timeout{x_2}} \!     (q_0, \emptyset)
        \! \xrightarrow{0.5} \!                     (q_0, \emptyset).
        \end{split}
  \end{gather*}
 The run \(\pi\) has three blocks \(B_1 = (i, \nilTimerKilled)\) ($x_1$ is started by $i$ and then discarded while its value is zero), \(B_2 = (i ~ \timeout{x_2}, \bot)\), and \(B_3 = (i ~ \timeout{x_1}, \nonNilTimerKilled)\) (\(x_1\) is again started in $B_3$ but $\pi$ ends before \(x_1\) reaches value zero). Those blocks are represented in \Cref{fig:ex:runs:pi}, where we visually represent the timer fate of $B_1$ (resp. $B_3$) by a dotted line finished by $\nilTimerKilled$ (resp. $\nonNilTimerKilled$). Finally, the run $\tau$ has its blocks depicted in \Cref{fig:ex:runs:tau}. This time, \(x_1\) is discarded before reaching zero, i.e., \(B_1 = (i, \nonNilTimerKilled)\).
\end{example}

As illustrated by the previous example, blocks satisfy the following property.

\begin{lemma}
  Let $\rho = (q_0,\kappa_0) ~ d_1 ~ i_1/u_1 ~ \dots ~ d_n ~ i_n/u_n ~ d_{n+1}  ~(q,\kappa)$ be a \maximal timed run of an \AwT. Then, the sequences of the blocks of $\rho$ form a partition of the set of indices $\{1, \dots, n\}$ of the actions of $\rho$.
\end{lemma}

Along a timed run of an \AwT \(\A\), it can happen that a timer times out at the same time that another action takes place.
This leads to a sort of \emph{nondeterminism}, as \(\A\) can process those concurrent actions in any order. This situation appears in \Cref{ex:blocks} each time a gap appears in the time lines of \Cref{fig:ex:runs}.
We call these situations \emph{races} that we formally define as follows. 

\begin{definition}[Race]\label{def:race}
Let \(B, B'\) be two blocks of a \maximal timed run \(\rho\) with timer fates \(\timerFate\) and \(\timerFate'\).  We say that
\(B\) and \(B'\) \emph{participate in a race} if:
\begin{itemize}
  \item either there exist actions \(i \in B\) and \(i' \in B'\) such that the sum of the delays between \(i\) and \(i'\) in $\rho$ is equal to zero, i.e., no time elapses between them,
  \item or there exists an action \(i \in B\) that is the first action along \(\rho\) to discard the timer started by the last action \(i' \in B'\) and \(\timerFate' = \nilTimerKilled\), i.e., the timer of \(B'\) (re)started by $i'$ reaches value zero when \(i\) discards it.
\end{itemize}
We also say that the actions $i$ and $i'$ participate in this race.
\end{definition}

The first case of the race definition appears in \Cref{fig:ex:runs:rho}, while the second case appears in \Cref{fig:ex:runs:pi} (see the race in which blocks $B_1$ and $B_3$ participate). The nondeterminism is highlighted in \Cref{fig:ex:runs:rho,fig:ex:runs:sigma} where two actions ($i$ and $\timeout{x}$) occur at the same time but are processed in a different order in each figure. Unfortunately, imposing a particular way of resolving races (i.e.\ imposing a particular action order) may seem arbitrary when modelling real-world systems. It is therefore desirable for the set of sequences of actions along timed runs to be independent to the resolution of races. 

\begin{definition}[\Robust]\label{def:robust}
  An \AwT $\A$ is \emph{\robust} iff for all \maximal timed runs $\rho \in \mtruns{\A}$ with races, there exists some $\rho' \in \mtruns{\A}$ with no races such that $\untime{\rho'} = \untime{\rho}$.
\end{definition}

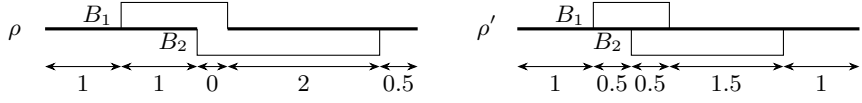
\begin{figure}[t]
    \centering
    \begin{subfigure}{.48\textwidth}
      \centering
      \input{figures/blocks/delays/rho.tex}
    \end{subfigure}
    \begin{subfigure}{.48\textwidth}
      \centering
      \input{figures/blocks/delays/rho_prime.tex}
    \end{subfigure}
    \caption{Modifying the delays in order to remove a race.}%
    \label{fig:move}\label{fig:rho_prime}
\end{figure}
 
\begin{example}\label{ex:robust}
Let us come back to the timed run $\rho$ of \Cref{ex:trun} that contains a race (see \Cref{fig:ex:runs:rho}).
By moving the second occurrence of action \(i\) slightly earlier in \(\rho\), we obtain the timed run \(\rho'\):
  \begin{align*}
    &\rho'\! ={} \!
    (q_0, \emptyset)
    \! \xrightarrow{1} \!                       (q_0, \emptyset)
    \! \xrightarrow[(x_1, 1)]{i} \!             (q_1, x_1\! =\! 1)
    \! \xrightarrow{0.5} \!                     (q_1, x_1\! =\! 0.5)
    \! \xrightarrow[(x_2, 2)]{i} \!             (q_2, x_1\! =\! 0.5, x_2\! =\! 2)\\
    &\! \xrightarrow{0.5} \!                    (q_2, x_1\! =\! 0, x_2\! =\! 1.5)
    \! \xrightarrow[\bot]{\timeout{x_1}} \!     (q_3, x_2\! =\! 1.5)
    \! \xrightarrow{1.5} \!                     (q_3, x_2\! =\! 0)
    \! \xrightarrow[\bot]{\timeout{x_2}} \!     (q_0, \emptyset)
    \! \xrightarrow{1} \!                       (q_0, \emptyset).
  \end{align*}
Notice that \(\untime{\rho'} = q_0 ~ i ~ q_1 ~ i ~ q_2 ~ \timeout{x_1} ~ q_3 ~ \timeout{x_2} ~ q_0 = \untime{\rho'}\). Moreover, $\rho'$ contains no races as indicated in \Cref{fig:move}.
\end{example}

Notice that several blocks could participate in the same race. The notion of block has been defined for \maximal timed runs only,
as we do not want to consider runs that end \emph{abruptly} during a race (some pending timeouts may not be processed at the end of the run, for instance). Moreover, it is always possible for the first delay to be positive as no timer is active in the initial state. Finally, non-zero delays at the start and the end of the runs allow to move blocks as introduced in \Cref{ex:robust} and further detailed in \Cref{sec:wiggle_run}.

In \Cref{sec:robust}, we study whether it is decidable that an \AwT is \robust,
and how to eliminate races in a \robust \AwT while keeping the same traces.
Before, we study the (classical) reachability problem in \Cref{sec:reachability}.  
 
\section{Reachability}\label{sec:reachability}\label{subsec:regionAut}

The \emph{reachability problem} asks, given an \AwT $\A$ and a state $q$, whether there exists a timed run $\rho \in \truns{\A}$ from the initial configuration
to some configuration $(q,\valuation)$. In this section, we argue that this problem is \PSPACE-complete. 

\begin{restatable}{theorem}{thmReachability}\label{thm:reachability}
    The reachability problem for \AwTs is \PSPACE-complete.
\end{restatable}

For hardness, we reduce from the acceptance problem for linear-bounded Turing machines (LBTM, for short), as done for timed automata, see e.g.~\cite{DBLP:journals/jlp/AcetoL02}.
In short, given an LBTM \(\M\) and a word \(w\) of length \(n\), we construct an \AwT that uses
\(n\) timers \(x_i\), \(1 \leq i \leq n\), such that the timer \(x_i\) encodes
the value of the \(i\)-th cell of the tape of \(\M\).
We also rely on a timer \(x\) that is always (re)started at one, and is used to
synchronize the \(x_i\) timers and the simulation of \(\M\).
The simulation is split into \emph{phases}: The \AwT first seeks the symbol
on the current cell \(i\) of the tape (which can be derived from the moment at
which the timer \(x_i\) times out, using the number of times \(x\) timed out
since the beginning of the phase).
Then, the \AwT simulates a transition of \(\M\) by restarting \(x_i\), reflecting
the new value of the \(i\)-th cell.
Finally, the \AwT can reach a designated state iff \(\M\) is in an accepting state.
Therefore, the reachability problem is \PSPACE-hard.

For membership, we follow the classical argument used to establish that the reachability problem for timed automata is in \PSPACE: We first define \emph{region automata} for \AwTs (which are a simplification of region automata for timed automata) and observe that
reachability in an \AwT reduces to reachability in the corresponding region automaton. The region automaton is of size exponential in the number of timers and polynomial in the number of states of the \AwT. Hence, the reachability problem for \AwTs is in \PSPACE{} via standard arguments.

We define region automata for \AwTs much like they are defined for timed
automata~\cite{alur1999timed,DBLP:journals/tcs/AlurD94,DBLP:books/daglib/0020348}.
Let $\A = (X,I,Q,q_0,\chi,\delta)$ be an \AwT.
For a timer \(x \in X\), \(c_x\) denotes the largest constant to which \(x\) is updated in \(\A\).
Let $C = \max_{x \in X} c_x$.
Two valuations \(\valuation\) and \(\valuation'\) are said
\emph{timer-equivalent}, noted \(\valuation \regionRelation \valuation'\), iff
$\dom{\valuation} = \dom{\valuation'}$ and the following hold for
all $x_1,x_2 \in \dom{\valuation}$:
\iflongversion
\begin{itemize}
\else
\begin{enumerate*}[label={\emph{(\roman*)}}]
\fi
  \item \(\lfloor \valuation(x_1) \rfloor = \lfloor \valuation'(x_1) \rfloor\),
  \item \(\mathrm{frac}(\valuation(x_1)) = 0\) iff
    \(\mathrm{frac}(\valuation'(x_1)) = 0\),
  \item \(\mathrm{frac}(\valuation(x_1)) \leq \mathrm{frac}(\valuation(x_2))\)
    iff \(\mathrm{frac}(\valuation'(x_1)) \leq
    \mathrm{frac}(\valuation'(x_2))\).
\iflongversion
\end{itemize} 
\else
\end{enumerate*} 
\fi
A \emph{timer region} for \(\A\) is an equivalence class of timer valuations
induced by \(\regionRelation\). We lift the relation to configurations:
$(q,\kappa) \regionRelation (q',\kappa')$ iff $\kappa
\regionRelation \kappa'$ and $q = q'$.
Finally, \(\regionClass{(q,\valuation)}\) denotes the equivalence class of $(q,\valuation)$.

We are now able to define a finite automaton called
the \emph{region automaton} of \(\A\) and denoted $\regionAutomaton$.
The alphabet of $\regionAutomaton$ is $\Sigma =  \{\delaySymbol\} \cup
\hat{I}$ where \(\delaySymbol\) is a special symbol used in non-zero delay transitions.
Formally, $\regionAutomaton$ is the finite automaton \((\Sigma, S, s_0, \Delta)\) where:
\begin{itemize}
  \item \(S = {\{(q,\valuation) \mid q \in Q, \valuation \in
    \Val{\chi(q)}\}}_{/\regionRelation}\), i.e., the quotient of the
    configurations by $\regionRelation$, is the set of states,
  \item \(s_0 = (q_0, \regionClass{\valuation_0})\) with
    \(\valuation_0\) the empty valuation, is the initial state,
  \item the set of transitions \(\Delta \subseteq S \times \Sigma \times S\)
    includes
    $(\regionClass{(q,\valuation)},\delaySymbol,\regionClass{(q,\valuation')})$
    if $(q,\valuation) \xrightarrow{d} (q,\valuation')$ in $\A$ whenever \(d > 0\),
    and
    $(\regionClass{(q,\valuation)},i,\regionClass{(q',\valuation')})$
    if $(q,\valuation) \xrightarrow[u]{i} (q',\valuation')$ in $\A$.
\end{itemize}

\noindent
It is easy to check that the timer-equivalence relation on configurations is a
\emph{(strong) time-abstracting
bisimulation}~\cite{DBLP:reference/mc/2018,DBLP:journals/fmsd/TripakisY01}.
That is, for all $(q_1,\kappa_1) \regionRelation (q_2,\kappa_2)$ the following
holds:
\begin{itemize}
  \item if $(q_1,\kappa_1) \xrightarrow[u]{i} (q'_1,\kappa'_1)$, then there is 
    $(q_2,\kappa_2) \xrightarrow[u]{i} (q'_2,\kappa'_2)$ with $(q'_1,\kappa'_1) \regionRelation (q'_2,\kappa'_2)$,
  \item if $(q_1,\kappa_1) \xrightarrow{d_1} (q_1,\kappa'_1)$, then there exists 
    $(q_2,\kappa_2) \xrightarrow{d_2} (q_2,\kappa'_2)$ where $d_1$,
    $d_2 > 0$ may differ such that $(q_1,\kappa'_1) \regionRelation (q_2,\kappa'_2)$, and
  \item the above also holds if $(q_1,\kappa_1)$ and $(q_2,\kappa_2)$ are swapped.
\end{itemize}
Using this property, we can prove the following about $\regionAutomaton$.
\begin{restatable}{lemma}{lemRegionAuto}\label{lem:region-auto}
  Let $\A = (X,I,Q,q_0,\chi,\delta)$ be an \AwT and $\regionAutomaton$ be its
  region automaton.
  \begin{enumerate}
    \item The size of $\regionAutomaton$ is linear in $|Q|$ and exponential in
      $|X|$. That is, $|S|$ is smaller than or equal to $|Q| \cdot |X|! \cdot
      2^{|X|} \cdot {(C+1)}^{|X|}$.
    \item There is a timed run $\rho$ of $\A$ that begins in $(q,\valuation)$ and ends in
      $(q',\valuation')$ iff there is a run $\rho'$ of $\regionAutomaton$ that
      begins in $\regionClass{(q,\valuation)}$ and ends in
      $\regionClass{(q',\valuation')}$.
  \end{enumerate}
\end{restatable}

\begin{example}\label{ex:timed-auto}
  Let us consider the \AwT \(\A\) of \Cref{fig:AwT} and the timed run \(\pi\) given in \Cref{ex:blocks}. The corresponding run $\pi'$ in the region automaton \(\regionAutomaton\) is
  \begin{align*}
    &(q_0, \regionClass{\emptyset})
    \! \xrightarrow{\delaySymbol} \!    (q_0, \regionClass{\emptyset})
    \! \xrightarrow{i} \!               (q_1, \regionClass{x_1 = 1})
    \! \xrightarrow{i} \!               (q_2, \regionClass{x_1 = 1, x_2 = 2})\\
    &\! \xrightarrow{\tau} \!           (q_2, \regionClass{x_1 = 0, x_2 = 1})
    \! \xrightarrow{i} \!               (q_2, \regionClass{x_1 = 1, x_2 = 1})
    \! \xrightarrow{\tau} \!           (q_2, \regionClass{x_1 = 0, x_2 = 0})\\
    &\! \xrightarrow{\timeout{x_2}} \!   (q_1, \regionClass{x_1 = 0})
    \! \xrightarrow{\timeout{x_1}} \!   (q_1, \regionClass{x_1 = 1})
    \! \xrightarrow{\tau} \!            (q_1, \regionClass{0 < x_1 < 1}).
  \end{align*}
  Notice that the transitions with delay zero of $\pi$ do not appear in $\pi'$.
\end{example}

\section{\Robust \AwTs}\label{sec:robust}

In this section, we study whether an \AwT \(\A\) is \robust, i.e., whether for
any \maximal timed run $\rho$ of $\A$ with races, there exists another run $\rho'$
with no races such that $\untime{\rho} = \untime{\rho'}$.
We are able to prove the next theorem.

\begin{restatable}{theorem}{thmRobust}\label{thm:robust}
Deciding whether an \AwT is \robust is \PSPACE-hard and in \THREEEXP{}. It is in \PSPACE{} if the sets of actions $I$ and of timers $X$ are fixed.
\end{restatable}

Our approach is, given $\rho \in \mtruns{\A}$, to study how to slightly move blocks along the time line of $\rho$ in a way to get another $\rho' \in \mtruns{\A}$ where the races are eliminated while keeping the actions in the same order as in $\rho$. We call this action \emph{wiggling}. Let us first give an example and then formalize this notion.

\begin{example}\label{ex:wiggling}
We consider again the \AwT of \Cref{fig:AwT}. We have seen in \Cref{ex:robust} and \Cref{fig:move} that the block $B_2$ of $\rho$ can be slightly moved to the left to obtain the timed run $\rho'$ with no race such that $\untime{\rho} = \untime{\rho'}$. \Cref{fig:move} illustrates how to move \(B_2\) by changing some of the delays.

In contrast, this is not possible for the timed run $\pi$ of \Cref{ex:blocks}. Indeed looking at \Cref{fig:ex:runs:pi}, we see that it is impossible to move block $B_2$ to the left due to its race with $B_1$ (remember that we need to keep the same action order). It is also not possible to move it to the right due to its race with $B_3$. Similarly, it is impossible to move $B_1$ neither to the right (due to its race with $B_2$), nor to the left (otherwise its timer will time out instead of being discarded by $B_3$). Finally, one can also check that block $B_3$ cannot be moved.   
\end{example}

Given a \maximal timed run $\rho = (q_0,\kappa_0) ~ d_1\! ~ i_1/u_1~\dots ~ d_{n} ~ i_n/u_n ~ d_{n+1} (q,\kappa) \in \mtruns{\A}$ and a block $B = (k_1 \ldots k_m,\timerFate)$ of $\rho$ participating in a race, we say that we can \emph{wiggle} $B$ if for some $\epsilon$, we can move \(B\) to the left (by $\epsilon < 0$) or to the right (by $\epsilon > 0$), and obtain a run \(\rho' \in \mtruns{\A}\) such that \(\untime{\rho} = \untime{\rho'}\) and \(B\) no longer participates in any race. More precisely, we have $\rho' = (q_0,\kappa_0) ~ d'_1 ~ i_1/u_1~\dots ~ d'_{n} ~ i_n/u_n ~ d'_{n+1} ~ (q,\kappa')$ such that 
\begin{itemize}
\item for all $i_{k_\ell} \in B$ with ${k_\ell} > 1$, if $i_{k_\ell - 1} \not\in B$ (the action before $i_{k_\ell}$ in $\rho$ does not belong to $B$), then $d'_{k_\ell} = d_{k_\ell} + \epsilon$,
\item if there exists $i_{k_\ell} \in B$ with $k_{\ell} = 1$ (the first action of $B$ is the first action of $\rho$), then $d'_1 = d_1 + \epsilon$,
\item for all $i_{k_\ell} \in B$ with ${k_\ell} < n$, if $i_{k_\ell + 1} \not\in B$ (the action after $i_{k_\ell}$ in $\rho$ does not belong to $B$), then $d'_{k_\ell+1} = d_{k_\ell+1} - \epsilon$,
\item if there exists $i_{k_\ell} \in B$ with $k_{\ell} = n$ (the last action of $B$ is the last action of $\rho$), then $d'_{n+1} = d_{n+1} - \epsilon$,
\item for all other $d'_k$, we have $d'_k = d_k$.
\end{itemize}
As $\rho' \in \mtruns{\A}$ and $\untime{\rho} = \untime{\rho'}$, we must have $d'_k \geq 0$ for all $k$ and $d'_1, d'_{n+1} > 0$.

We say that we can wiggle $\rho$, or that $\rho$ is \emph{wigglable}, if it is possible to wiggle its blocks, \emph{one block at a time}, to obtain $\rho' \in \mtruns{\A}$ with no races such $\untime{\rho} = \untime{\rho'}$. Hence if all \maximal timed runs with races of an \AwT $\A$ are wigglable, then $\A$ is \robust.

In the next sections, we first associate a graph with any $\rho \in \mtruns{\A}$ in a way to characterize when $\rho$ is wigglable thanks to this graph. We then state the equivalence between the \robust characteristic of an \AwT and the property that all $\rho \in \mtruns{\A}$  can be wiggled (\Cref{prop:RobustCharacterization}).
This allows us to provide logic formulas to determine whether an AT has an
unwigglable run, and then to prove the upper bound of \Cref{thm:robust}. We also discuss its lower bound. Finally, we discuss some sufficient hypotheses for a \robust AT\@.

\subsection{Wiggling a run}\label{sec:wiggle_run}

In this section, given an \AwT $\A$ and a \maximal timed run $\rho \in \mtruns{\A}$, we study the conditions required to be able to wiggle $\rho$. For this purpose, we define the following graph $G_\rho$ associated with $\rho$.
When two blocks \(B\) and \(B'\) of $\rho$ participate in a race, we write \(B \prec B'\) if there exist actions \(i \in B\) and \(i' \in B'\) such that $i,i'$ participate in this race and, according to \Cref{def:race}: 
\begin{itemize}
\item either $i$ occurs before $i'$ along $\rho$ and the total delay between $i$ and $i'$ is zero.
\item or the timer of \(B'\) (re)started by $i'$ reaches value zero when \(i\) discards it.
\end{itemize}
We define the \emph{block graph} \(\blockGraph[\rho] = (V, E)\) of \(\rho\) where \(V\) is the set of blocks of \(\rho\) and $E$ has an edge \((B, B')\) iff \(B \prec B'\).

\begin{example}\label{ex:block_graph}
  \begin{figure}[t]
    \centering
    \begin{subfigure}{.45\textwidth}
      \centering
      \input{figures/blockgraph/rho.tex}
      \caption{Graph $\blockGraph[\rho]$.}%
      \label{fig:ex:block_graph:rho}
    \end{subfigure}
    \begin{subfigure}{.45\textwidth}
      \centering
      \input{figures/blockgraph/pi.tex}
      \caption{Graph $\blockGraph[\pi]$.}%
      \label{fig:ex:block_graph:pi}
    \end{subfigure}
    \caption{Block graphs of the timed runs \(\rho\) and \(\pi\) of \Cref{ex:blocks}.}%
    \label{fig:ex:block_graph}
  \end{figure}
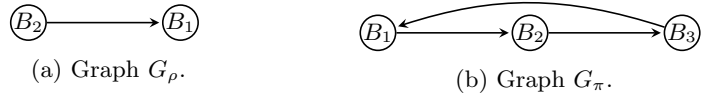

  Let \(\A\) be the \AwT from \Cref{fig:AwT}, and \(\rho\) and \(\pi\) be the timed runs from \Cref{ex:blocks}, whose block decomposition is represented in \Cref{fig:ex:runs:rho,fig:ex:runs:pi}. For the run \(\rho\), it holds that \(B_2 \prec B_1\) leading to the block graph \(\blockGraph[\rho]\) depicted in \Cref{fig:ex:block_graph:rho}. For the run $\pi$, we get the block graph \(\blockGraph[\pi]\) depicted in~\Cref{fig:ex:block_graph:pi}.

  Notice that \(\blockGraph[\rho]\) is acyclic while \(\blockGraph[\pi]\) is cyclic. By the following proposition, this difference is enough to characterize
  that \(\rho\) can be wiggled and \(\pi\) cannot.
\end{example}

\begin{restatable}{proposition}{propWigglableRun}\label{prop:wigglable_run}
  Let $\A$ be an \AwT and \(\rho \in \mtruns{\A}\) be a \maximal timed run with races. Then, $\rho$ can be wiggled iff \(\blockGraph[\rho]\) is acyclic.
\end{restatable}

Intuitively, a block cannot be moved to the left (resp.\ right) if it has a
predecessor (resp.\ successor) in the block graph, due to the races in which it participates.
Hence, if a block has both a predecessor and a successor, it cannot be wiggled
(see \Cref{fig:ex:runs:pi,fig:ex:block_graph:pi} for instance).
Then, the blocks appearing in a cycle of the block graph cannot be wiggled.
The other direction holds by observing that we can do a topological sort of the
blocks if the graph is acyclic. We then wiggle the blocks, one by one, according to that sort.

The next corollary will be useful in the following sections. It is illustrated by \Cref{fig:races} with the simple cycle \((B_0,B_1,\allowbreak B_2,B_3,B_4,B_0)\).

\begin{restatable}{corollary}{corCycle}\label{cor:cycle}
  Let $\A$ be an \AwT and \(\rho \in \mtruns{\A}\) be a \maximal timed run with races. Suppose that \(\blockGraph[\rho]\) is cyclic. Then there exists a cycle $\mathcal C$ in \(\blockGraph[\rho]\) such that
  \begin{itemize}
  \item any block of $\mathcal C$ participate in exactly two races described by this cycle,
  \item for any race described by $\mathcal C$, exactly two blocks participate in the race,
  \item the blocks $B = (k_1 \ldots k_m,\timerFate)$ of $\mathcal C$ satisfy either $m \geq 2$, or $m=1$ and $\timerFate = \nilTimerKilled$.
  \end{itemize}
\end{restatable}

\begin{figure}[t]
  \centering
  \input{figures/blocks/abstractForMSO.tex}
  \caption{Races of a \maximal timed run \(\rho\) with $B_\ell \prec B_{\ell + 1 \bmod 5}$, $0 \leq \ell \leq 4$.}%
  \label{fig:races}
\end{figure}
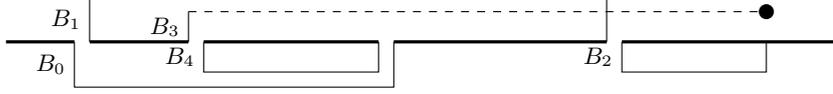

From the definition of wiggling, we know that if all \maximal timed runs with races of an \AwT $\A$ are wigglable, then $\A$ is \robust. The converse also holds as stated in the next theorem. By \Cref{prop:wigglable_run}, this means that an \AwT is \robust iff the block graph of all its \maximal timed run is acyclic.

\begin{restatable}{theorem}{propRobustCharacterization}\label{prop:RobustCharacterization}
  An \AwT $\A$ is \robust
  \begin{itemize}
  \item iff any \maximal timed run $\rho \in \mtruns{\A}$ with races can be wiggled,
  \item iff for any \maximal timed run $\rho \in \mtruns{\A}$, its graph $\blockGraph[\rho]$ is acyclic.
  \end{itemize}
\end{restatable}

Let us sketch the missing proof. By modifying \(\rho \in \mtruns{\A}\) to explicitly encode when a timer is discarded, one can show the races of $\rho$ cannot be avoided if the block graph of $\rho$ is cyclic as follows. Given two actions \(i, i'\) of this modified run, it is possible to define the \emph{\elapsed} between \(i\) and \(i'\), noted \(\distance(i, i')\), from the sum \(d\) of all delays between \(i\) and \(i'\): if \(i\) occurs before \(i'\),
then \(\distance(i, i') = d\), otherwise \(\distance(i, i') = -d\).
Lifting this to a sequence of actions from $\rho$ is defined naturally.
Then, one can observe that the \elapsed of a cyclic sequence of actions is zero,
i.e., \(\distance(i_1, i_2, \dotsc, i_k, i_1) = 0\).
Finally, from a cycle of $\blockGraph[\rho]$ as described in \Cref{cor:cycle}, we extract a cyclic sequence of actions and prove, thanks to the concept of \elapsed, that any run \(\rho'\) such that \(\untime{\rho} = \untime{\rho'}\) must contain some races.

\subsection{Existence of an unwigglable run}\label{subsec:MSO}

In this section, we give the intuition for the announced complexity bounds for the
problem of deciding whether an \AwT \(\A\) is \robust (\Cref{thm:robust}). 

Let us begin with the \THREEEXP{}-membership. The crux of our approach is to use the characterization of the \robust property given in
\Cref{prop:RobustCharacterization}, to work with a slight modification of the
region automaton $\regionAutomaton$ of $\A$,
and to construct a finite-state automaton whose language is the set of runs of \(\regionAutomaton\) whose block graph is cyclic. Hence deciding whether $\A$ is \robust amounts to deciding whether the language accepted by the latter automaton is empty. To do so, we construct a \emph{monadic second-order} (MSO, for short; see~\cite{gradel2003automata,DBLP:reference/hfl/Thomas97} for an introduction) formula that is satisfied by words \(w\) labeling a run \(\rho\) of \(\regionAutomaton\) iff the block graph of $\rho$ is cyclic.

Our \emph{modification} of $\regionAutomaton$ is best seen as additional annotations
on the states and transitions of $\regionAutomaton$. We extend the
alphabet $\Sigma$ as follows:
\begin{enumerate*}[itemjoin={{; }}, label={\emph{(\roman*)}}]
  \item we add a timer to each action $i \in \hat{I}$ to remember the updated
    timers
  \item we also use new symbols \(\killSymbol{x}\), $x \in X$, with the intent of explicitly encoding in $\regionAutomaton$ when the timer \(x\) is  discarded while its value is zero.
\end{enumerate*}
Therefore, the modified alphabet is $\Sigma =  \{\delaySymbol\} \cup (\hat{I} \times
\hat{X}) \cup \{\killSymbol{x} \mid x \in X \}$ where \(\hat{X} = X \cup
\{\bot\}\). As a transition in $\A$ can discard more than one timer, we store
the set $D$ of discarded timers in the states of $\regionAutomaton$, as well as outgoing
transitions labeled by $\killSymbol{x}$, for all discarded timers $x$. For
this, the states of $\regionAutomaton$ become \(S = \{(q, \regionClass{\valuation}, D) \mid q \in Q, \valuation \in \Val{\chi(q)}, D \subseteq X \}\) and $\Delta$ is modified in
the natural way so that $D$ is updated as required. Note that the size of this
modified  $\regionAutomaton$ is only larger than what is stated in
\Cref{lem:region-auto} by a factor of $2^{|X|}$.

Note that any $x$-block $(i_{k_1}, \ldots, i_{k_m},\timerFate)$ of a timed run $\rho$ in $\A$ is translated into the sequence of symbols $(i'_{k_1},\ldots,i'_{k_m},\timerFate')$ in the corresponding run $\rho'$ of the modified $\regionAutomaton$ with an optional symbol $\timerFate'$ such that:
\begin{itemize}
    \item $i'_{k_\ell} = (i_{k_\ell},x)$, for $1 \leq \ell < m$,
    \item $i'_{k_m} = (i_{k_m},\bot)$ if $\timerFate = \bot$, and $(i_{k_m},x)$ otherwise,
    \item $\timerFate' = \killSymbol{x}$ if $\timerFate = \nilTimerKilled$, and $\timerFate'$ does not exist otherwise.
\end{itemize}
It follows that, instead of considering \maximal timed runs $\rho \in \mtruns{\A}$ and their block graph $\blockGraph[\rho]$, we work with their corresponding (padded) runs, blocks, and block
graphs in the modified region automaton $\regionAutomaton$ of $\A$.

\begin{example}\label{ex:timed-auto:modified}
  The run \(\pi'\) of \Cref{ex:timed-auto} becomes
  \begin{align*}
    (q_0&, \regionClass{\emptyset}, \emptyset)
    \! \xrightarrow{\delaySymbol} \!            (q_0, \regionClass{\emptyset}, \emptyset)
    \! \xrightarrow{(i, x_1)} \!                (q_1, \regionClass{x_1\! =\! 1}, \emptyset)
    \! \xrightarrow{(i, x_2)} \!                (q_2, \regionClass{x_1\! =\! 1, x_2\! =\! 2}, \emptyset)\\
    &\! \xrightarrow{\tau} \!                   (q_2, \regionClass{x_1\! =\! 0, x_2\! =\! 1}, \emptyset)
    \! \xrightarrow{(i, x_1)} \!                (q_2, \regionClass{x_1\! =\! 1, x_2\! =\! 1}, \{x_1\})\\
    &\! \xrightarrow{\killSymbol{x_1}} \!       (q_2, \regionClass{x_1\! =\! 1, x_2\! =\! 1}, \emptyset)
    \! \xrightarrow{\tau} \!                    (q_2, \regionClass{x_1\! =\! 0, x_2\! =\! 0}, \emptyset)\\
    &\! \xrightarrow{(\timeout{x_2}, \bot)} \!  (q_1, \regionClass{x_1\! =\! 0}, \emptyset)
    \! \xrightarrow{(\timeout{x_1}, x_1)} \!    (q_1, \regionClass{x_1\! =\! 1}, \emptyset)
    \! \xrightarrow{\tau} \!                    (q_1, \regionClass{0\! <\! x_1\! <\! 1}, \emptyset).
  \end{align*}
  The transition with label $\killSymbol{x_1}$ indicates that the timer $x_1$ is discarded in the original timed run while its value equals zero (see the race in which blocks $B_1$ and $B_3$ participate in \Cref{fig:ex:runs:pi}).
  The three blocks of $\pi$ become $B'_1 = ((i, x_1),\killSymbol{x_1})$, $B'_2 = ((i, x_2),(\timeout{x_2}, \bot))$, and $B'_3 = ((i, x_1),(\timeout{x_1}, x_1))$ in $\pi'$. The fact that in $\pi$, $B_1$ and $B_2$ participate in a race (with a zero-delay between their respective actions $i$ and $i$), is translated in $\pi'$ with the non existence of the $\tau$-symbol between the symbols $(i, x_1)$ and $(i, x_2)$ in $B'_1$ and $B'_2$ respectively.
\end{example}

\begin{restatable}{lemma}{lemEncoding}\label{lem:encoding}
  Let \(\A\) be an \AwT and \(\regionAutomaton\) be its modified region automaton. We can construct an MSO formula \(\nowigglingMSO\) of size linear in $I$ and $X$ such that a word labeling a run $\rho$ of $\regionAutomaton$ satisfies $\nowigglingMSO$ iff $\rho$ is a \maximal run that cannot be wiggled. Moreover, the formula $\nowigglingMSO$, in prenex normal form, has three quantifier alternations.
\end{restatable}

The formula $\nowigglingMSO$ of this lemma describes the existence of a cycle $\mathcal C$ of blocks \(B_0, B_1, \dotsc, B_{k-1}\) such that \(B_{\ell} \prec B_{\ell + 1 \bmod k}\) for any \(0\leq\ell\leq k-1\), as described in~\Cref{cor:cycle} (see \Cref{fig:races}). To do so, we consider the actions (i.e., symbols of the alphabet $\Sigma$ of $\regionAutomaton$) participating in the races of $\mathcal C$: \(i_0, i_1, \dotsc, i_{k-1}\) and \(i'_0, i'_1, \dotsc, i'_{k-1}\) such that for all $\ell$, \(i_\ell, i'_\ell\) belong to the same block, and $i'_\ell,i_{\ell + 1 \bmod k}$ participate in a race. One can write MSO formulas expressing that two actions participate in a race (there is no $\tau$ transition between them), that two actions belong to the same block, and, finally, the existence of these two action sequences.

From the formula $\nowigglingMSO$ of \Cref{lem:encoding}, by the B\"uchi-Elgot-Trakhtenbrot theorem, we can construct a finite-state automaton whose language is the set of all words satisfying $\nowigglingMSO$. Its size is triple-exponential. We then compute the intersection $\N$ of this automaton with \(\regionAutomaton\) --- itself exponential in size. Finally, the language of $\N$ is empty iff each \maximal timed run of $\A$ can be wiggled, and this can be checked in polynomial time with respect to the triple-exponential size of $\N$, thus showing the \THREEEXP{}-membership of \Cref{thm:robust}. Notice that when we fix the sets of inputs $I$ and of timers $X$, as the formula \(\nowigglingMSO\) becomes of constant size, the automaton $\N$ has now exponential size. Checking its emptiness can be done \enquote{on the fly}, yielding a nondeterministic decision procedure which requires a polynomial space only. We thus obtain that, under fixed inputs and timers, deciding whether an \AwT is \robust is in \PSPACE.

The complexity lower bound of \Cref{thm:robust} follows from the \PSPACE-hardness of the reachability problem for \AwTs (see the intuition given in \Cref{sec:reachability}).
We can show that any run in the \AwT constructed from the given LBTM and word $w$ can be wiggled.
Once the designated state for the reachability reduction is reached, we add a
widget that forces a run that cannot be wiggled.
Therefore, as the only way of obtaining a run that cannot be wiggled is to
reach a specific state (from the widget), the problem whether an \AwT is \robust is \PSPACE-hard. Notice that this hardness proof is no longer valid if we fix the sets $I$ and $X$.

\subsection{Sufficient hypotheses}

Let us discuss some sufficient hypotheses for an AT $\A$ to be \robust.
\begin{enumerate}
    \item If every state in $\A$ has at most one active timer, then $\A$ is \robust. Up to renaming the timers, we actually have a single-timer AT in this case.
    \item If we modify the notion of timed run by imposing non-zero delays everywhere in the run, then $\A$ is \robust. Indeed, the only races that can appear are when a zero-valued timer is discarded, and it is impossible to form a cycle in the block graph with only this kind of races. Imposing a non-zero delay before a timeout is debateable. Nevertheless, imposing a non-zero delay before inputs only is not a sufficient hypothesis (we have counter-examples). 
    \item Let us fix a total order $<$ over the timers, and modify the semantics of an AT to enforce that, in a race, any action of an \(x\)-block is processed before an action of a \(y\)-block, if $x < y$ ($x$ is preemptive over $y$). Then the AT is \robust. Towards a contradiction, assume there are blocks \(B_0, B_1, \dotsc, B_{k-1}\) forming a cycle as described in \Cref{cor:cycle}, where each \(B_i\) is an \(x_i\)-block. By the order and the races, we get \(x_0 \leq x_1 \leq \dotsc \leq x_{k-1} \leq x_0\), i.e., we have a single timer (as in the first hypothesis). Hence, it is always possible to wiggle, which is a contradiction.
\end{enumerate}

\section{Conclusion and future work}
In this paper, we studied automata with timers. We proved that the reachability problem for \AwTs is \PSPACE-complete. Moreover, given a \maximal timed run in an \AwT, we defined a decomposition of its actions into blocks, and provided a way to remove races (concurrent actions) inside the run by wiggling blocks one by one. We also proved that this notion of wiggling is necessary and sufficient to decide whether an \AwT is \robust. Finally, we showed that deciding whether an \AwT is \robust is in \THREEEXP{} and \PSPACE-hard.

For future work, it may be interesting to tighten the complexity bounds for the latter decision problem, both when fixing the sets $I$ and $X$ and when not. A second important direction, which we plan to pursue, is
to work on a learning algorithm for \AwTs, as initiated in~\cite{DBLP:conf/lata/VaandragerB021} with Mealy machines with one timer. This would allow us to construct \AwTs from real-world systems, such as network protocols, in order to verify that these systems behave as expected.

\bibliographystyle{splncs04}
\bibliography{refs}

\clearpage
\appendix

\input{appendix}

\end{document}

%% file: intro.tex
\section{Introduction}

Timed automata were introduced by Alur \& Dill~\cite{DBLP:journals/tcs/AlurD94} as finite-state automata equipped with real-valued clock variables for measuring the time between state transitions. These clock variables all increase at the same rate when time elapses, they can be reset along transitions, and be used in guards along transitions and in invariant predicates for states.
Timed automata have become a framework of choice for modeling and analysis of real-time systems, equipped with a rich theory, supported by powerful tools, and with numerous applications~\cite{DBLP:reference/mc/BouyerFLMO018}.

Interestingly, whereas the values of clocks in a timed automaton \emph{increase} over time, designers of real-time systems (e.g.\ embedded controllers and network protocols) typically use timers to enforce timing constraints, and the values of these timers \emph{decrease} over time.
If an application starts a timer with a certain value $t$, then this value decreases over time and after $t$ time units --- when the value has become $0$ --- a timeout event occurs.
It is straightforward to encode the behavior of timers using a timed automaton.
Timed automata allow one to express a richer class of behaviors than what can be described using timers, and can for instance express that the time between two events is contained in an interval $[t-d, t+ d]$.
Moreover, timed automata can express constraints on the timing between arbitrary events, not just between start and timeout of timers.

However, the expressive power of timed automata entails certain problems.
For instance, one can easily define timed automata models in which time stops at some point (timelocks) or an infinite number of discrete transitions occurs in a finite time (Zeno behavior). Thus timed automata may describe behavior that cannot be realized by any physical system.
Also, learning~\cite{DBLP:journals/iandc/Angluin87,DBLP:conf/dagstuhl/HowarS16} of timed automata models in a black-box setting turns out to be challenging~\cite{DBLP:conf/concur/GrinchteinJP06,DBLP:journals/tcs/GrinchteinJL10,AnCZZZ20}.
For a learner who can only observe the external events of a system and their timing, it may be really difficult to infer the logical predicates (invariants and guards) that label the states and transitions of a timed automaton model of this system.
As a result, all known learning algorithms for timed automata suffer from combinatorial
explosions, which severely limits their practical usefulness.

For these reasons, it is interesting to consider variations of timed automata whose expressivity is restricted by using timers instead of clocks.
Vaandrager et al.\ \cite{DBLP:conf/lata/VaandragerB021} study deterministic
Mealy machines with a single timer (MM1T). In an MM1T, a transition may start a
timer by setting it to a certain constant. Whenever a timer reaches zero, it
produces an observable timeout symbol that triggers a transition in the
automaton.  Vaandrager et al.\ provide a black-box active learning algorithm for MM1Ts,
and evaluate it on a number of realistic applications, showing that it outperforms the timed automata based approaches of Aichernig et al.\ \cite{AichernigPT20} and An et al.\ \cite{AnCZZZ20}.
However, whereas MM1Ts only support a single timer, the genetic programming approach of~\cite{AichernigPT20} is able to learn models with multiple clocks/timers.

If we want to extend the learning algorithm of~\cite{DBLP:conf/lata/VaandragerB021} to a setting with multiple timers, we need to deal with the issue of \emph{races}, i.e., situations where multiple timers reach zero (and thus timeout) simultaneously.
If a race occurs, then (despite the automaton being deterministic!) the automaton can process the simultaneous timeouts in various orders, leading to nondeterministic behavior.
This means that during learning of an automaton with multiple timers, a learner needs to offer the inputs at specific times in order to avoid the occurrence of races. As long as there are no races, the behavior of the automaton will be deterministic, and a learner may determine, for each timeout, by which preceding input it was caused by slightly \emph{wiggling} the timing of inputs and check whether the timing timeout changes in a corresponding manner. 

\paragraph{Contribution}
In this work,
we take the one-timer definition
from~\cite{DBLP:conf/lata/VaandragerB021} and extend it to multiple timers
while --- to avoid overcomplicating the model --- keeping the restriction that
every transition can start or restart at most one timer. We first study the
complexity of the configuration reachability problem for this model and
establish that it is \PSPACE-complete.
Then, we turn our attention to the problem of determining whether it is possible to wiggle the delays between the inputs in a run, in a way to avoid races.
The importance of the latter is twofold. First, automata with timers may not be an attractive modelling formalism in the presence of behaviors that do not align with those of the real-world systems they are meant to abstract. Second, the absence of races is a key property used in the learning algorithm for automata with a single timer. In this direction, we provide an effective characterization of when an automaton is \robust and establish that the related decision problem is in \THREEEXP{} and \PSPACE-hard.
In a more pragmatic direction, while again leveraging our characterization, we show that with fixed input and timer sets, the problem is in \PSPACE. Finally, we also give some simple yet sufficient conditions for an automaton to be \robust.

%% file: figures/definitions/MMT.tex
\begin{tikzpicture}[
    automaton,
    node distance = 1cm and 2.5cm,
    state/.append style = {
        inner sep = 2pt,
        minimum size = 0pt,
    }
]
    \node [state, initial]          (q0)    {\(q_0\)};
    \node [state, right=of q0]      (q1)    {\(q_1\)};
    \node [state, below=of q1]      (q2)    {\(q_2\)};
    \node [state, below=of q0]      (q3)    {\(q_3\)};

    \path
        (q0)    edge                node [above]    {\(i,(x_1, 1)\)}
                    (q1)
        (q1)    edge [bend left]    node            {\(i, (x_2, 2)\)}
                    (q2)
                edge [loop above]   node            {\(\timeout{x_1}, (x_1, 1)\)}
                    (q1)
        (q2)    edge [loop right]   node            {\(i, (x_1, 1)\)}
                    (q2)
                edge             
                                    node [below]    {\(\timeout{x_1},\bot\)}
                    (q3)
                edge [bend left]   node             {\(\timeout{x_2}, \bot\)}
                    (q1)
        (q3)    edge [loop left]   node            {\(i, (x_2, 1)\)}
                    (q3)
                edge                node            {\(\timeout{x_2}, \bot\)}
                    (q0)
    ;
\end{tikzpicture}

%% file: figures/blocks/rho.tex
\begin{tikzpicture}[xscale=0.6]
    \coordinate (startB1)   at (0.5, 0);
    \coordinate (1toB1)     at (1.7, 0);
    \coordinate (startB2)   at (1.5, 0);
    \coordinate (1toB2)     at (3.7, 0);
    \coordinate (end)       at (4, 0);

    \foreach \s/\t/\h/\l in {startB1/1toB1/0.35/B_1, startB2/1toB2/-0.35/B_2} {
        \draw[block]
            let
                \p{s} = (\s),
                \p{t} = (\t)
            in
                (\p{s})
                    -- node [left=-0.1] {\(\l\)} (\x{s}, \h)
                    -- (\x{t}, \h)
                    -- (\p{t})
        ;
    }

    \draw[timeline]
        (0, 0) -- (startB2)
        (1toB1) -- (1toB2)
        (1toB2) -- (end)
    ;
    \draw[hole]
        (startB2) -- (1toB1)
    ;
\end{tikzpicture}

%% file: figures/blocks/sigma.tex
\begin{tikzpicture}[xscale=0.6]
    \coordinate (startB1)   at (0.5, 0);
    \coordinate (1toB1)     at (1.5, 0);
    \coordinate (2toB1)     at (2.7, 0);
    \coordinate (startB2)   at (1.7, 0);
    \coordinate (1toB2)     at (3.7, 0);
    \coordinate (end)       at (4, 0);

    \foreach \s/\t/\h/\l in {startB1/1toB1/0.35/B_1, 1toB1/2toB1/0.35/,
        startB2/1toB2/-0.35/B_2} {
        \draw[block]
            let
                \p{s} = (\s),
                \p{t} = (\t)
            in
                (\p{s})
                    -- node [left=-0.1] {\(\l\)} (\x{s}, \h)
                    -- (\x{t}, \h)
                    -- (\p{t})
        ;
    }

    \draw[timeline]
        (0, 0) -- (1toB1)
        (startB2) -- (1toB2)
        (1toB2) -- (end)
    ;
    \draw[hole]
        (1toB1) -- (startB2)
    ;
\end{tikzpicture}

%% file: figures/blocks/pi.tex
\begin{tikzpicture}[xscale=0.7]
    \coordinate (startB1)   at (0.5, 0);
    \coordinate (1toB1)     at (1.7, 0);
    \coordinate (startB2)   at (0.7, 0);
    \coordinate (1toB2)     at (2.7, 0);
    \coordinate (startB3)   at (1.7, 0);
    \coordinate (1toB3)     at (2.9, 0);
    \coordinate (2toB3)     at (3.2, 0);
    \coordinate (end)       at (3.2, 0);
    
    \foreach \s/\t/\h/\l in {startB2/1toB2/-0.35/B_2, startB3/1toB3/0.35/B_3} {
        \draw[block]
            let
                \p{s} = (\s),
                \p{t} = (\t)
            in
                (\p{s})
                    -- node [left=-0.1] {\(\l\)} (\x{s}, \h)
                    -- (\x{t}, \h)
                    -- (\p{t})
        ;
    }

    \foreach \s/\t/\h/\type/\l in {
        startB1/1toB1/0.5/nullTimerKilled/B_1,
        1toB3/2toB3/0.35/nonNullTimerKilled/
    } {
        \draw[block]
            let
                \p{s} = (\s),
            in
                (\p{s})
                    -- node [left=-0.1] {\(\l\)} (\x{s}, \h)
        ;
        \draw[blockBeforeKilled]
            let
                \p{s} = (\s),
                \p{t} = (\t),
            in
                (\x{s}, \h)
                    -- (\x{t}, \h)
                node[\type] at (\x{t}, \h) {}
        ;
    }

    \draw[timeline]
        (0, 0) -- (startB1)
        (startB2) -- (1toB2)
        (1toB3) -- (end)
    ;
    \draw[hole]
        (startB1) -- (startB2)
        (1toB2) -- (1toB3)
    ;
\end{tikzpicture}

%% file: figures/blocks/tau.tex
\begin{tikzpicture}[xscale=0.75]
    \coordinate (startB1)   at (0.5, 0);
    \coordinate (1toB1)     at (1.2, 0);
    \coordinate (startB2)   at (0.7, 0);
    \coordinate (1toB2)     at (2.7, 0);
    \coordinate (startB3)   at (1.2, 0);
    \coordinate (1toB3)     at (2.2, 0);
    \coordinate (end)       at (3, 0);
    
    \foreach \s/\t/\h/\l in {startB2/1toB2/-0.35/B_2, startB3/1toB3/0.35/B_3} {
        \draw[block]
            let
                \p{s} = (\s),
                \p{t} = (\t)
            in
                (\p{s})
                    -- node [left=-0.1] {\(\l\)} (\x{s}, \h)
                    -- (\x{t}, \h)
                    -- (\p{t})
        ;
    }

    \foreach \s/\t/\h/\type/\l in {
        startB1/1toB1/0.5/nonNullTimerKilled/B_1
    } {
        \draw[block]
            let
                \p{s} = (\s),
            in
                (\p{s})
                    -- node [left=-0.1] {\(\l\)} (\x{s}, \h)
        ;
        \draw[blockBeforeKilled]
            let
                \p{s} = (\s),
                \p{t} = (\t),
            in
                (\x{s}, \h)
                    -- (\x{t}, \h)
                node[\type] at (\x{t}, \h) {}
        ;
    }

    \draw[timeline]
        (0, 0) -- (startB1)
        (startB2) -- (1toB2)
        (1toB2) -- (end)
    ;
    \draw[hole]
        (startB1) -- (startB2)
    ;
\end{tikzpicture}

%% file: figures/blocks/delays/rho.tex
\begin{tikzpicture}
    \node [text centered] at (-0.4, 0) {$\rho$\strut};

    \coordinate (startB1)   at (1, 0);
    \coordinate (1toB1)     at (2.4, 0);
    \coordinate (startB2)   at (2, 0);
    \coordinate (1toB2)     at (4.4, 0);
    \coordinate (end)       at (4.9, 0);

    \foreach \s/\t/\h/\l in {startB1/1toB1/0.35/B_1, startB2/1toB2/-0.35/B_2} {
        \draw[block]
            let
                \p{s} = (\s),
                \p{t} = (\t)
            in
                (\p{s})
                    -- node [left] {\(\l\)} (\x{s}, \h)
                    -- (\x{t}, \h)
                    -- (\p{t})
        ;
    }

    \draw[timeline]
        (0, 0) -- (startB2)
        (1toB1) -- (end)
    ;
    \draw[hole]
        (startB2) -- (1toB1)
    ;

    \coordinate (offset) at (0, -0.5);

    \draw [delay]
        (offset)
            -- node [below] {1}
                ($(startB1) + (offset)$)
    ;
    \draw [delay]
        ($(startB1) + (offset)$)
            -- node [below] {1}
                ($(startB2) + (offset)$)
    ;
    \draw [delay]
        ($(startB2) + (offset)$)
            -- node [below] {0}
                ($(1toB1) + (offset)$)
    ;
    \draw [delay]
        ($(1toB1) + (offset)$)
            -- node [below] {2}
                ($(1toB2) + (offset)$)
    ;
    \draw [delay]
        ($(1toB2) + (offset)$)
            -- node [below] {\(0.5\)}
                ($(end) + (offset)$)
    ;
\end{tikzpicture}

%% file: figures/blocks/delays/rho_prime.tex
\begin{tikzpicture}
    \node [text centered] at (-0.4, 0) {$\rho'$\strut};

    \coordinate (startB1)   at (1, 0);
    \coordinate (1toB1)     at (2, 0);
    \coordinate (startB2)   at (1.5, 0);
    \coordinate (1toB2)     at (3.5, 0);
    \coordinate (end)       at (4.5, 0);

    \foreach \s/\t/\h/\l in {startB1/1toB1/0.35/B_1, startB2/1toB2/-0.35/B_2} {
        \draw[block]
            let
                \p{s} = (\s),
                \p{t} = (\t)
            in
                (\p{s})
                    -- node [left] {\(\l\)} (\x{s}, \h)
                    -- (\x{t}, \h)
                    -- (\p{t})
        ;
    }

    \draw[timeline]
        (0, 0) -- (end)
    ;

    \coordinate (offset) at (0, -0.5);

    \draw [delay]
        (offset)
            -- node [below] {1}
                ($(startB1) + (offset)$)
    ;
    \draw [delay]
        ($(startB1) + (offset)$)
            -- node [below] {$0.5$}
                ($(startB2) + (offset)$)
    ;
    \draw [delay]
        ($(startB2) + (offset)$)
            -- node [below] {$0.5$}
                ($(1toB1) + (offset)$)
    ;
    \draw [delay]
        ($(1toB1) + (offset)$)
            -- node [below] {$1.5$}
                ($(1toB2) + (offset)$)
    ;
    \draw [delay]
        ($(1toB2) + (offset)$)
            -- node [below] {$1$}
                ($(end) + (offset)$)
    ;
\end{tikzpicture}

%% file: figures/blockgraph/rho.tex
\begin{tikzpicture}[
    blockGraph,
    block/.style = {
        state,
        inner sep = 0pt,
        minimum size = 0pt,
    }
]
    \node [block]               (B2)    {\(B_2\)};
    \node [block, right=of B2]  (B1)    {\(B_1\)};

    \path
        (B2)    edge    (B1)
    ;
\end{tikzpicture}

%% file: figures/blockgraph/pi.tex
\begin{tikzpicture}[
    blockGraph,
    block/.style = {
        state,
        inner sep = 0pt,
        minimum size = 0pt,
    },
]
    \node [block]                       (B1)    {\(B_1\)};
    \node [block, right=of B1]          (B2)    {\(B_2\)};
    \node [block, right=of B2]          (B3)    {\(B_3\)};
 
    \path
        (B1)    edge    (B2)
        (B2)    edge    (B3)
        (B3)    edge [bend right=18]    (B1)
    ;
\end{tikzpicture}

%% file: figures/blocks/abstractForMSO.tex
\begin{tikzpicture}
    \coordinate (startB1)   at (0.9, 0);
    \coordinate (1toB1)     at (5.1, 0);

    \coordinate (startB2)   at (1.1, 0);
    \coordinate (1toB2)     at (7.9, 0);

    \coordinate (startB3)   at (8.1, 0);
    \coordinate (1toB3)     at (10, 0);

    \coordinate (startB4)   at (2.4, 0);
    \coordinate (1toB4)     at (10, 0);

    \coordinate (startB5)   at (2.6, 0);
    \coordinate (1toB5)     at (4.9, 0);
    
    \foreach \s/\t/\h/\l in {
        startB1/1toB1/-0.6/B_0,
        startB2/1toB2/0.6/B_1,
        startB3/1toB3/-0.4/B_2,
        startB5/1toB5/-0.4/B_4
    } {
        \draw[block]
            let
                \p{s} = (\s),
                \p{t} = (\t)
            in
                (\p{s})
                    -- node [left] {\(\l\)} (\x{s}, \h)
                    -- (\x{t}, \h)
                    -- (\p{t})
        ;
    }

    \foreach \s/\t/\h/\l in {startB4/1toB4/0.4/B_3} {
        \draw[block]
            let
                \p{s} = (\s),
            in
                (\p{s})
                    -- node [left] {\(\l\)} (\x{s}, \h)
        ;
        \draw[blockBeforeKilled]
            let
                \p{s} = (\s),
                \p{t} = (\t),
            in
                (\x{s}, \h)
                    -- (\x{t}, \h)
                node[nullTimerKilled] at (\x{t}, \h) {}
        ;
    }

    \draw[timeline]
        (0, 0) -- (startB1)
        (startB2) -- (startB4)
        (startB5) -- (1toB5)
        (1toB1) -- (1toB2)
        (startB3) -- (11, 0)
    ;
    \draw[hole]
        (startB1) -- (startB2)
        (startB4) -- (startB5)
        (1toB5) -- (1toB1)
        (1toB2) -- (1toB3)
    ;
\end{tikzpicture}

%% file: appendix.tex
\section{Proof of PSPACE lower bound of \Cref{thm:reachability}}\label{sec:reachability:lower}

\thmReachability*

A \emph{linear-bounded Turing machine} (LBTM, for short) $\M = (\Sigma,Q,q_0,q_F,T)$ is a nondeterministic Turing machine which can only use $|w|$ cells of the tape to determine whether an input $w$ is accepted. Formally, $\Sigma$ is a finite alphabet, $Q$ is a finite set of states, $q_0$ and $q_F$ are the initial and final states respectively, and $T \subseteq Q \times \Sigma \times \Sigma \times \{L,R\} \times Q$ is the set of transitions. A configuration of $\M$ is
a triple $(q,w,i) \in Q \times \Sigma^* \times \mathbb{N}^{>0}$ where $q$ denotes the current control state, $w = w_1 \ldots w_n$ is the content of the tape, and $i$ is the position of the tape head. We say a transition $(q, \alpha, \alpha', D, q') \in T$ is enabled in a configuration $(q,w,i)$ if $w_i = \alpha$. In that case, taking the transition results in the machine reaching the new configuration $(q',w',i')$ with $w'_i = \alpha'$, $w'_j = w_j$ for all $j \neq i$, and $i' = i + 1$ if $D = R$ or $i' = i - 1$ otherwise. A given word $w$ is said to be accepted by $\M$ if there is a sequence of transitions from $(q_0,w,1)$ to a configuration of the form $(q_F,w',i)$.  Deciding whether a given LBTM accepts a given word is \PSPACE-complete~\cite{DBLP:books/aw/HopcroftU79}.

\begin{proof}[of Theorem~\ref{thm:reachability} (lower bound)]
We show that the acceptance problem for LBTMs can be reduced in polynomial time to the reachability problem for \AwTs. The proof is inspired by the one presented by Aceto and Laroussinie for the fact that reachability is \PSPACE-hard for timed automata~\cite[Section 3.1]{DBLP:journals/jlp/AcetoL02}.
	
Let $\M = (\Sigma, Q, q_0, q_F, T)$ be an LBTM and let $w \in \Sigma^*$ be an input word with $|w| = n$.  From $\M$ and $w$ we are going to build, in polynomial time, an \AwT $\A_{\M, w}$ such that $w$ is accepted by $\M$ iff there exists a timed run that ends in a specific state $r_\done$ of $\A_{\M, w}$.\footnote{In the following, we rely on the notion of blocks to pass on the intuition. While blocks are only defined on \maximal timed runs, it is easy to check that any timed run reaching $r_\done$ can be turned into a \maximal run.}

Let $\Sigma = \{ a_1 ,\dots, a_k \}$. Then, every $2k+2$ time units the \AwT will simulate a single step of the LBTM\@. This is what we call a \emph{phase}. 

Let $T = \{ t_1,\ldots, t_m \}$ be the set of transitions of $\M$.  Then, $\A_{\M, w}$ has inputs $I = \{ \go \} \cup T$ and timers $X = \{x,x_1,\ldots,x_n\}$. Now, a state of $\A_{\M,w}$ is an element of $\{ r_0, \dots, r_n, r_\done, r_\sink \}$ (with $r_0$ the initial state) or a tuple
  \(
  \langle \state, i, \tsymbol, \clock \rangle,
  \)
where:
  \begin{itemize}
    \item $\state \in Q$ records the current state of the LBTM,
    \item $i \in \{ 1, \dots, n \}$ records the current position of the tape head,
    \item $\tsymbol \in \{0 ,\dots, k \}$ records the index of the last symbol read from the tape ($0$ when no symbol has been read or the last read symbol has been processed),
    \item $\clock \in \{ 0, 1, \ldots, 2k+1 \}$.
  \end{itemize}

The \AwT starts with an initialization phase in which it goes through states $r_0$ to $r_n$ to start all the timers via a sequence of $\go$-inputs. Execution begins with a $\go$-input that starts timer $x$:
  \(
    r_0 \xrightarrow[(x, 1)]{\go} r_1.
  \)
In order to reach state $r_\done$, all the other timers need to be started
before timer $x$ times out. We use timer $x_i$, for $1 \leq i \leq n$, to
record the value of the $i$-th tape cell: if this value is symbol $a_j \in
\Sigma$, then timer $x_i$ will time out 
when $\lfloor
\clock/2 \rfloor = j$
(it will become clear
later how this is possible). Timer $x_i$ is started in state $r_i$ and set to
its appropriate value: 
\begin{itemize}
    \item if $i < n$, we have the transition
    \(
      r_i \xrightarrow[(x_i, 2j)]{\go} r_{i+1}
    \)
    that initializes timer $x_i$ with value $2j$,
    \item if $i=n$, we have the transition
    \(
    r_n \xrightarrow[(x_n, 2j)]{\go} \langle q_0, 1, 0, 0 \rangle
    \)
    that initializes timer $x_n$ with the same value $2j$ and starts the computation of the LBTM\@.
\end{itemize}
All timeout transitions from $r_i$, $1 \leq i \leq n$, go to $r_\sink$. This ensures that --- in order to reach state $r_\done$--- all timers $x_i$ are initialized.
Hence, all timed runs in $\A_{\M, w}$ that reach $r_\done$ have the same starting blocks, as depicted in Figure~\ref{fig:run_prefix}: an $x$-block $B_x$ and $n$ $x_i$-blocks $B_{x_i}$.\footnote{Notice that some of these blocks participate in a race if their timers are initialized at the same time. For instance, all timers can be started at the same time $d = 0$.}

\begin{figure}[t]
    \centering
    \input{figures/blocks/pspace-hard/init.tex}
    \caption{Starting blocks.}%
    \label{fig:run_prefix}
\end{figure}
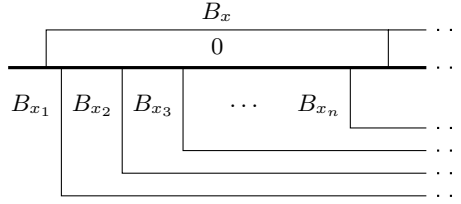

We use timer $x$ to advance the value of $\clock$ that runs cyclically from $0$ to $2k+1$:
  \begin{equation}\label{eqn:condition1}
    \begin{split}
    \clock > 0 \vee \tsymbol = 0 ~~\Rightarrow &\\ 
    \langle \state, i, \tsymbol,
    \clock \rangle \xrightarrow[(x, 1)]{\timeout{x}} {}& \langle \state, i,
    \tsymbol, (\clock + 1) \;\mathrm{mod}\; 2k+2 \rangle. 
    \end{split}
  \end{equation}
(It will become clear later why the condition on this transition is required.) When timer $x_\ell$ times out, for some $\ell \neq i$, then we just restart it so that it will times out at exactly the same point in the next phase:
  \[
    \langle \state, i, \tsymbol, \clock \rangle
    \xrightarrow[(x_\ell, 2k+2)]{\timeout{x_\ell}} \langle \state, i, \tsymbol, \clock \rangle.
  \]
When timer $x_i$ times out, then we restart it in the same way, but in addition we store the index of the symbol that it encodes in the state of the \AwT:
  \[
    \langle \state, i, \tsymbol, \clock \rangle \xrightarrow[(x_i,
  2k+2)]{\timeout{x_i}} \langle \state, i, \lfloor \clock/2 \rfloor , \clock \rangle.
  \]
In order to see why this is true, suppose timer $x_i$ has been started at time
$d$ with value $2j$. Then, $d \in [0,1]$ (see \Cref{fig:run_prefix}) and $x_i$ will expire at time $d' = d + 2j$, so $d'\in [2j, 2j+1]$.  At this time, the value of $\clock$ will be either $2j$ or $2j+1$, and thus $j = \lfloor \clock/2 \rfloor$.  

When the value of $\clock$ becomes $0$ again (a next phase begins), the LBTM $\M$ has read a symbol from the tape, so $\tsymbol > 0$, and $\M$ may (nondeterministically) take a transition.  For each transition $t = (q, \alpha, \alpha', q', L)$ of $\M$, with $\alpha = a_\tsymbol$ and $\alpha'= a_j$, the \AwT has a transition:
  \[
  \langle \state, i, \tsymbol, 0 \rangle \xrightarrow[(x_i, 2j)]{t} \langle \state', i-1, 0, 0 \rangle.
  \]
The \AwT also has transitions, mutatis mutandis, for each $t = (q, \alpha, \alpha', q', R)$ of the LBTM\@. In order to ensure that these transitions are taken before timer $x$ times out, we add transitions:
  \begin{equation} \label{eqn:condition2}
  \tsymbol > 0 ~~ \Rightarrow ~~\langle \state, i, \tsymbol, 0 \rangle \xrightarrow{\timeout{x}} r_\sink.
  \end{equation}
The condition $\clock > 0 \vee \tsymbol = 0$ from~\eqref{eqn:condition1} and the condition $\tsymbol > 0$ from~\eqref{eqn:condition2} both ensure that we see first action $t$ and then action $\timeout{x}$ in a timed run that reaches $r_\done$.

In \Cref{fig:phases}, we fix \(k = 2\), and, for a timed run $\rho$ that reaches $r_\done$, we indicate the sequence of phases, with the cyclic value of $\clock$ from $0$ to $2k+1$. We also indicate the block $B_x$ such that timer $x$ is restarted each time it times out along $\rho$. Finally, we indicate two $x_i$-blocks, $B^1_{x_i}$ and $B^2_{x_i}$, such that in $B^1_{x_i}$, timer $x_i$ is started in phase~$1$, restarted during this phase, and restarted again during phase~$2$, until it is discarded when $\clock = 0$ in phase~$3$; and in $B^2_{x_i}$, timer $x_i$ is started with a new value dictated by the processed transition of the LBTM\@. Note that there may be other blocks (for timers \(x_j\), with \(j \neq i\)) which are not represented in the figure.

\begin{figure}[t]
  \centering
  \input{figures/blocks/pspace-hard/run.tex}
  \caption{The beginning of a timed run that reaches $r_\done$ with \(k = 2\).}%
  \label{fig:phases}
\end{figure}
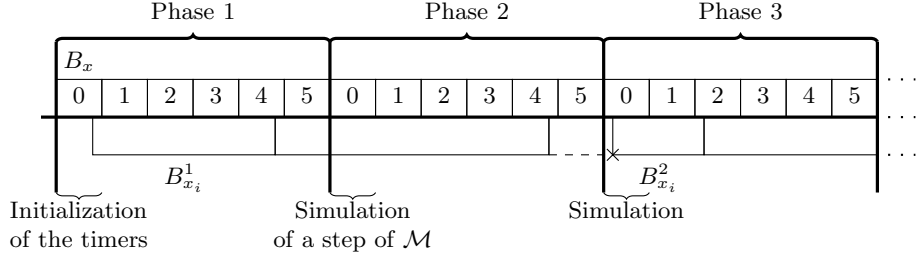

As soon as the LBTM reaches $q_F$, the \AwT may proceed to its final state $r_\done$:
  \[
    \langle \state_F, i, \tsymbol, \clock \rangle \xrightarrow{\go} r_\done.
  \]
For all states of the \AwT, outgoing transitions for actions that have not been specified lead to $r_\sink$. 

Finally, we define the active timers in each state of the \AwT as follows: $\chi(r_0) = \chi(r_\sink) = \chi(r_\done) = \emptyset$, $\chi(r_i) = \{x,x_1,x_2,\ldots,x_{i-1}\}$ for $1 \leq i \leq n$, and $\chi(r) =\{x,x_1,x_2,\ldots,x_n\}$ for all the other states $r$. 

It is clear that $\A_{\M, w}$ can be constructed from $\M$ and $w$ in polynomial time and that $r_\done$ is reachable in the \AwT iff the LBTM accepts $w$.
\qed\end{proof}

\section{Proof of \Cref{lem:region-auto}}

\lemRegionAuto*

\begin{proof}[of \Cref{lem:region-auto}]
For the first statement of the lemma, recall that each state of the region automaton is of the form $(q, \regionClass{\valuation})$.
The number of states $q$ is equal to $|Q|$. Concerning the number of region classes $\regionClass{\valuation}$, ${(C+1)}^{|X|}$ is related to the integer parts of the timers between 0 and $C$, $2^{|X|}$ to which timers have a zero fractional part, and $|X|! $ to the order of these fractional parts. 

The second statement of the lemma follows from the definition of the region automaton $\regionAutomaton$. Notice that delay transitions with delay $0$ in $\A$ disappear in $\regionAutomaton$.
\qed\end{proof}

\section{Proof of PSPACE upper bound of \Cref{thm:reachability}}

\thmReachability*

\begin{proof}[upper bound of \Cref{thm:reachability}]
To decide the reachability problem for \AwTs, by \Cref{lem:region-auto}, we can simulate a run of the corresponding region automaton. Instead of constructing the region automaton in full, we can do so ``on the fly''. This yields a nondeterministic decision procedure for the reachability problem which, due to the form $(q,\regionClass{\valuation})$ of the states of $\regionAutomaton$,
requires polynomial space only. Since $\NPSPACE = \PSPACE$, we obtain the upper bound stated in Theorem~\ref{thm:reachability}.
\qed\end{proof}

\section{Proof of \Cref{prop:wigglable_run}}

\propWigglableRun*

Before establishing this result, we prove the following intermediate result.

\begin{restatable}{lemma}{lemWiggleOneBlock}\label{lem:wiggle_one_block}
  Let \(\blockGraph[\rho]\) be the block graph of \(\rho \in \mtruns{\A}\) and \(B\) be a block in this graph. It is impossible to wiggle \(B\) iff \(B\) has at least one predecessor and at least one successor in \(\blockGraph[\rho]\).
\end{restatable}

\begin{proof}
  We prove the lemma by showing both directions.

\medskip
  \(\Leftarrow\)
  Suppose that \(B\) has a block \(B'\) as predecessor and \(B''\) as successor, that is, \(B' \prec B\) and \(B \prec B''\). Notice that it may be that \(B' = B''\). However, \(B \neq B'\) and \(B \neq B''\) by the definition of races. Let us prove that it is impossible to wiggle \(B\) by arguing that it is not feasible to move \(B\) to the right nor to the left. Given \(B' \prec B\), let us prove that we cannot move \(B\) to the left. We have two cases for the actions \(i \in B\) and \(i' \in B'\) that participate in the race:
  \begin{itemize}
    \item Action $i'$ occurs before action $i$ along $\rho$ and the sum of the delays
      between these two actions is zero. Thus the delay $d$ before $i$ in $\rho$ is equal to zero and it is impossible to have $d+\epsilon \geq 0$ for any $\epsilon <0$. This implies that no movement of $B$ to the left is possible.
    \item The timer $x$ of \(B\) (re)started by $i$ reaches value zero when \(i'\) discards it. By moving \(B\) to the left by some $\epsilon < 0$, $x$ times out and therefore an action \(\timeout{x}\) occurs, while it does not occur in \(\rho\) (as $i'$ discards $x$). As we want to keep the same untimed trace as for \(\rho\), it is impossible to move \(B\) to the left.
  \end{itemize}
  With symmetrical arguments, we obtain that we cannot move \(B\) to the right, as \(B \prec B''\). We conclude that we cannot wiggle \(B\).
  
\medskip
  \(\Rightarrow\)
  We prove this direction by contraposition. We first assume that \(B\) has no predecessor (however it may have successors $C$). We argue that \(B\) can be wiggled by moving it to the left. From the definition of a race, we obtain that:
  \begin{itemize}
    \item Since $\rho$ is a \maximal timed run, the first delay $d_1$ of $\rho$ is non-zero.
    \item For each action $i'$ before some action $i \in B$ such that $i' \not\in B$, the delay $d_{i'}$ between $i'$ and $i$ must be non-zero (as $B$ has no predecessor).
    \item The timer fate $\timerFate_B$ of \(B\) is either \(\bot\) or \(\nonNilTimerKilled\). Indeed, assume by contradiction that $\timerFate_B = \nilTimerKilled$. Recall that, by definition of a \maximal run, timers cannot have a zero value at the end of a run. Therefore, there must exist a block \(B''\) that discards the timer of \(B\) while its valuation is zero.
    Thus, we have that \(B'' \prec B\) which is not possible.
  \end{itemize}
 From these observations, we conclude that there is enough room to move \(B\) to the left. Indeed, it is possible to choose some $\epsilon < 0$ with $|\epsilon|$ small enough, such that $d_1 + \epsilon > 0$, $d_{i'} + \epsilon > 0$ for all the delays $d_{i'}$ mentioned above, and in a way that if $\timerFate_B = \nonNilTimerKilled$ then the new timer fate of $B$ is still equal to $\nonNilTimerKilled$. In this way we produce a timed run \(\rho'\) such that \(\untime{\rho} = \untime{\rho'}\).

  It remains to explain that the blocks $C$ participating in a race with $B$ in $\rho$ no longer participate in such a race in $\rho'$. Let $C$ be one of these blocks, hence $B \prec C$. We have again two cases:
    \begin{itemize}
    \item There exist $i \in B$ and $i' \in C$ such that $i$ occurs before $i'$ in $\rho$ and the total delay between them is zero. In the timed run $\rho'$, this delay becomes equal to $-\epsilon > 0$ and $B,C$ no longer participate in a race.
    \item The timer fate $\timerFate_C$ of \(C\) is equal to \(\nilTimerKilled\) and the timer of $C$ is discarded by $B$ along $\rho$. In $\rho'$, we get that $\timerFate_C = \nonNilTimerKilled$ and $B,C$ no longer participate in a race.
  \end{itemize}
  
  It follows that if $B$ has no predecessor, we can wiggle it. With symmetrical arguments, if \(B\) has no successor, we can also wiggle it. Hence, the lemma holds. 
\qed\end{proof}

Now, we proceed to proving \Cref{prop:wigglable_run}.

\begin{proof}[of \Cref{prop:wigglable_run}]
  We prove the equivalence by showing both directions.

\medskip
  \(\Rightarrow\)
  We prove this direction by contraposition. Suppose that \(\blockGraph[\rho]\) has a cycle that we can assume to be simple, i.e., there are \(k > 1\) distinct blocks $B_\ell$, $0 \leq \ell \leq k-1$ such that \(B_\ell \prec B_{\ell + 1 \bmod k}\). As every block \(B_\ell\) has a predecessor and a successor in this cycle, we cannot wiggle \(B_\ell\) by \Cref{lem:wiggle_one_block}. Thus, $\rho$ cannot be wiggled as it is impossible to resolve the races in which the blocks \(B_\ell\) participate.

\medskip
  \(\Leftarrow\)
  Assume \(\blockGraph[\rho]\) is acyclic. Hence, we can compute a topological sort of $\blockGraph[\rho]$ restricted to the blocks participating in the races of $\rho$. Let $B$ be the greatest block with respect to this sort, that is, $B$ has no successor and it has predecessors. By~\Cref{lem:wiggle_one_block}, we can wiggle $B$ by moving it slightly to the right, thus eliminating the races between \(B\) and all the other blocks. We thus obtain a new timed run \(\rho'\) such that \(\untime{\rho} = \untime{\rho'}\) and $\blockGraph[\rho']$ has the same vertices as $\blockGraph[\rho]$ and strictly less edges ($B$ becomes an isolated vertex). We repeat this process until the blocks of the graph are all isolated, meaning that the resulting timed run has no races and the same untimed trace as $\rho$.
\qed\end{proof}

\section{Proof of \Cref{cor:cycle}}

\corCycle*

\begin{proof}
As \(\blockGraph[\rho]\) is cyclic, we consider a cycle of minimal length. First notice that a block \(B\) of this cycle can only appear once per race. Indeed, the value at which a timer is (re)started in the block is positive, thus imposing non-zero delays between two actions of $B$ (i.e., two actions of $B$ can not participate in a common race). Second, by minimality of its length, the cycle is simple, implying that each of its blocks participates in exactly two races, one with its unique successor (in the cycle) and another one with its unique predecessor. Third, assume that three blocks \(B_1, B_2\) and \(B_3\) participate in a common race, in that order. By the previous remark, they are pairwise distinct, and it must be that \(B_1 \prec B_2\), \(B_2 \prec B_3\) and \(B_1 \prec B_3\). It follows that we get a smaller cycle by eliminating $B_2$, which is a contradiction. Finally, assume that the cycle contains some block $B = (k_1,\timerFate)$ with $\timerFate \neq \nilTimerKilled$. Let $B_1 \prec B$ (resp. $B \prec B_2$) be the predecessor (resp.\ successor) of $B$ in the cycle. Due to the form of $B$, the three blocks $B_1, B_2$ and $B$ participate in the same race, which is impossible. 
\qed\end{proof}

\section{Proof of \Cref{prop:RobustCharacterization}}

\propRobustCharacterization*

For the first equivalence, we only need to prove one implication of this result, as by definition of wiggling, an \AwT $\A$ is \robust if all its \maximal timed runs with races are wigglable. The second equivalence is then a consequence of \Cref{prop:wigglable_run}.

For this purpose, we need to introduce some new notions. Given a \maximal timed run $\rho = (q_0,\kappa_0) ~ d_1 ~ i_1/u_1 ~ \dots ~ d_n ~ i_n/u_n ~ d_{n+1} ~ (q,\kappa)$, we extend it with additional transitions indicating when a timer has been discarded in the following way.\footnote{In~\Cref{subsec:MSO}, this was done in the modified region automaton $\regionAutomaton$ of $\A$ with the new symbols $\killSymbol{x}$ indicating that the timer $x$ was discarded when its value was zero. We here also consider the case when $x$ is discarded with a non-zero value. To get more intuition about the proof of \Cref{prop:RobustCharacterization}, we recommend reading \Cref{proof:lem:encoding}.} Let $(q_{\ell-1},\kappa_{\ell-1}) ~ d_\ell ~ i_\ell/u_\ell ~ (q_{\ell},\kappa_{\ell})$ be any transition of $\rho$ such that the set $D = \{y_1, \ldots, y_m\}$ of timers discarded by $i_\ell$ is not empty. Then we replace $(q_{\ell},\kappa_{\ell})$ by the sequence of transitions 
\[
  (q_{\ell},\kappa_{\ell}) ~ 0 ~ j_1/\bot ~ (q_{\ell},\kappa_{\ell}) ~ 0 ~j_2/\bot ~ (q_{\ell},\kappa_{\ell}) ~\cdots~ (q_{\ell},\kappa_{\ell}) ~ 0 ~j_m/\bot ~(q_{\ell},\kappa_{\ell})
\] 
such that 
\begin{itemize}
\item for all $k$, $1 \leq k \leq m$, if $y_k$ was discarded by $i_\ell$ when its value was zero, then $j_k = \nilTimerKilled$, otherwise $j_k = \nonNilTimerKilled$,
\item each delay is zero, and
\item each update is $\bot$.
\end{itemize}
We denote by $\ext{\rho}$ the resulting extended run, such that symbols $\nilTimerKilled$ and $\nonNilTimerKilled$ are also called actions.

\begin{figure}[t]
  \centering
  \input{figures/blocks/extended/pi.tex}
  \caption{The extended run of $\pi$ and its block decomposition.}%
  \label{fig:extrun}
\end{figure}
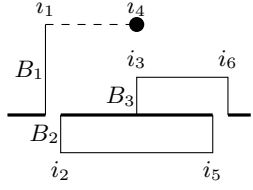

\begin{example}
    Let us come back to the timed run $\pi$ of \Cref{ex:blocks} (see also \Cref{fig:ex:runs:pi}). Its transition 
    $(q_2, x_1 = 1, x_2 = 2) ~1 ~i/(x_1, 1) ~ (q_2, x_1 = 1, x_2 = 1)$ discards timer $x_1$ when its value is zero. Therefore in $\ext{\pi}$, we have $(q_2, x_1 = 1, x_2 = 1)$ being replaced by $(q_2, x_1 = 1, x_2 = 1) ~0 ~\nilTimerKilled/\bot ~ (q_2, x_1 = 1, x_2 = 1)$. This extended run and its block decomposition are depicted in~\Cref{fig:extrun}, such that $i_1,i_2,\ldots,i_7$ is the sequence of actions along $\ext{\pi}$ with $i_4 = \nilTimerKilled$.
\end{example}

Let $\rho$ be a \maximal timed run and $\ext{\rho}$ be its extended run. Given two actions \(i\) and \(i'\) of $\ext{\rho}$, the \emph{\elapsed} between $i$ and $i'$, denoted by $\distance_\rho(i, i')$, is defined as follows from the sum $d$ of all delays between $i$ and $i'$ in $\ext{\rho}$: if \(i\) occurs before \(i'\), then \(\distance_\rho(i, i') = d\), otherwise \(\distance_\rho(i, i') = -d\). Notice that the \elapsed is sensitive to the order of the actions along $\ext{\rho}$, and if \(i\) and \(i'\) participate in a race, then \(\distance_\rho(i, i') = \distance_\rho(i', i) = 0\). 
We naturally lift this definition to a sequence of actions \(i_1, i_2, \dotsc, i_k\) as
\[
  \distance_\rho(i_1, i_2, \dotsc, i_k) = \sum_{\ell = 1}^{k-1} \distance_\rho(i_\ell, i_{\ell + 1}).
\]
The following lemma is trivial, as the \elapsed between two actions has a sign that depends on the relative position of the actions. It is illustrated by \Cref{ex:nullelapsedtime} below.

\begin{lemma}\label{lem:cycleNul}
If the sequence \(i_1, i_2, \dotsc, i_k\) is a cycle (that is, $k \geq 3$ and $i_k  = i_1$), then $\distance_\rho(i_1, i_2, \dotsc, i_k) = 0$.
\end{lemma}

\begin{example}\label{ex:nullelapsedtime}  
We consider again the timed run $\pi$ and its extended run $\ext{\pi}$. Recall that $\pi$ cannot be wiggled as $\blockGraph[\pi]$ is cyclic (see \Cref{fig:ex:block_graph:pi}). From this cycle and the block decomposition of $\ext{\pi}$ (see \Cref{fig:extrun}), we extract the following sequence of actions: $i_1,i_2,i_5,i_6,i_3,i_4,i_1$. Notice that it is a cycle such that any two consecutive actions are either in the same block, or participate in a race, and are enumerated in a way to \enquote{follow} the cycle $B_1 \prec B_2 \prec B_3 \prec B_1$ of $\blockGraph[\pi]$. For instance, the first two actions $i_1,i_2$ describes the race $B_1 \prec B_2$, then $i_2,i_5$ both belong to $B_2$, then $i_5,i_6$ describes the race $B_2 \prec B_3$, etc. We have $\distance_{\pi}(i_1,i_2,i_5,i_6,i_3,i_4,i_1) = 0 + 2 + 0 - 1 + 0 - 1  = 0$.
\end{example}

We now proceed to the proof of \Cref{prop:RobustCharacterization}.

\begin{proof}[of \Cref{prop:RobustCharacterization}]
The second equivalence holds by \Cref{prop:wigglable_run}. Let us focus on proving that \(\A\) is \robust iff any \maximal timed run \(\rho \in \mtruns{\A}\) is wigglable. As, by definition, it is obvious that \(\A\) is \robust if any \maximal timed run is wigglable, we show the other direction.

Towards a contradiction, assume \(\A\) is \robust and there exists \(\rho_1 \in \mtruns{\A}\) with races that is not wigglable. Since \(\A\) is \robust, there exists another \maximal timed run \(\rho_2\) without races and
  such that \(\untime{\rho_1} = \untime{\rho_2}\). We consider the two extended runs $\ext{\rho_1}$ and $\ext{\rho_2}$.

By \Cref{prop:wigglable_run}, there must exist a cycle \(\mathcal{C}\) in the block graph of $\rho_1$. We assume that $\mathcal C$ is as described in \Cref{cor:cycle} and we study it on $\ext{\rho_1}$ (instead of $\rho_1$). That is, \(\mathcal{C}\) is composed of \(k > 1\) distinct blocks \(B_\ell\), \(0 \leq \ell \leq k -1\), such that \(B_\ell \prec B_{\ell + 1 \bmod k}\), and
  \begin{itemize}
      \item any block $B_\ell$ participates in exactly two races described by $\mathcal{C}$,
      \item for any race described by $\mathcal{C}$, exactly two blocks participate in the race,
      \item the blocks of the cycle have at least two actions (of which one can be $\nilTimerKilled$).
  \end{itemize}
We thus have the following sequence of actions from \(\ext{\rho_1}\)
\[
  \mathcal{S}_1 = i'_0,i_1,i'_1, \ldots, i_{k-1},i'_{k-1},i_{0}, i'_0
\]
that is a cycle and such that for all $\ell$, $0 \leq \ell < k$ (see also~\Cref{ex:nullelapsedtime}):
\begin{itemize}
\item $i_\ell$ and $i'_\ell$ are the two symbols of $B_\ell$ that participate in (different) races of \(\mathcal{C}\),
\item $i'_\ell \in B_\ell$ and $i_{\ell + 1 \bmod k} \in B_{\ell + 1 \bmod k}$ participate in a race of $\mathcal C$, that is, $\distance_{\rho_1}(i'_\ell,i_{\ell + 1 \bmod k}) = 0$,
\item $i'_{\ell}$ occurs before $i_{\ell + 1 \bmod k}$ in $\ext{\rho_1}$ (since $B_\ell \prec B_{\ell + 1 \bmod k}$).
\end{itemize}
By \Cref{lem:cycleNul}, we have $\distance_{\rho_1}(\mathcal{S}_1) = 0$. Therefore, 
\begin{eqnarray}\label{eq:distanceCycle}
  \distance_{\rho_1}(\mathcal{S}_1) = \sum_{\ell = 0}^{k-1} \distance_{\rho_1}(i_\ell,i'_\ell) = 0.
\end{eqnarray}

Let us now study $\ext{\rho_2}$ knowing that \(\untime{\rho_1} = \untime{\rho_2}\). Both \maximal timed runs \(\rho_1\) and \(\rho_2\) (and thus \(\ext{\rho_1}\) and \(\ext{\rho_2}\)) must have the same block decomposition. Indeed recall that \(\A\) is deterministic and we see the same actions. Hence, it must be that \(\rho_1\) and \(\rho_2\) follow the same transitions, with the same updates alongside both runs. We then have the same sequences of triggered actions, and therefore the same block decomposition in both runs.

We can thus consider the blocks of $\mathcal{C}$ seen as blocks in $\ext{\rho_2}$, and the sequence $\mathcal{S}_2 = j'_0,j_1,j'_1 \ldots, j_{k-1}, j'_{k-1}, j_{0},j'_0$ from $\ext{\rho_2}$ that corresponds to the sequence $\mathcal{S}_1$ from $\ext{\rho_1}$. We have the following properties for all $\ell, 0 \leq \ell < k$:
\begin{itemize}
\item if $i_\ell \in \hat{I}$ (resp. $i'_\ell \in \hat{I}$), then $j_\ell = i_\ell$ ($j'_\ell = i'_\ell$),
\item if $i_\ell = \nilTimerKilled$ (resp. $i'_\ell = \nilTimerKilled$), then $j_\ell = \nonNilTimerKilled$ ($j'_\ell = \nonNilTimerKilled$), as $\rho_2$ has no races and \(\untime{\rho_1} = \untime{\rho_2}\),
\item if $i_\ell, i'_\ell \in \hat{I}$, then $\distance_{\rho_1}(i_\ell,i'_\ell) = \distance_{\rho_2}(j_\ell,j'_\ell)$, as $i_\ell = j_\ell, i'_\ell = j'_\ell$ are in the same block in both $\ext{\rho_1}$ and $\ext{\rho_2}$,
\item if one among $i_\ell, i'_\ell$ is equal to $\nilTimerKilled$, then $|\distance_{\rho_1}(i_\ell,i'_\ell)| > |\distance_{\rho_2}(j_\ell,j'_\ell)|$ as the corresponding action in $\ext{\rho_2}$ is equal to $\nonNilTimerKilled$ (the timer has been discarded earlier in $\rho_2$ than in $\rho_1$),
\item we have $\distance_{\rho_2}(j'_\ell,j_{\ell + 1 \bmod k}) \neq 0$, as $\rho_2$ has no races,
\item moreover, $\distance_{\rho_2}(j'_\ell,j_{\ell + 1 \bmod k}) > 0$ because $i'_{\ell}$ occurs before $i_{\ell + 1 \bmod k}$ in $\ext{\rho_1}$ and $\untime{\rho_1} = \untime{\rho_2}$.
\end{itemize}

Let us first assume that no action $\nilTimerKilled$ ever appears in $\mathcal{S}_1$. Hence, there is no action $\nonNilTimerKilled$ in $\mathcal{S}_2$. As $\mathcal{S}_2$ is a cycle, by~\Cref{lem:cycleNul}, we have that $\distance_{\rho_2}(\mathcal{S}_2) = 0$. It follows by~\eqref{eq:distanceCycle} that
\begin{align} \label{eq:contradiction}
    0 = \distance_{\rho_2}(\mathcal{S}_2) 
    &= \sum_{\ell = 0}^{k-1} \distance_{\rho_2}(j'_\ell, j_{\ell + 1 \bmod
    k}) + \sum_{\ell = 0}^{k-1} \distance_{\rho_2}(j_\ell, j'_\ell)\nonumber \\
    &> 0 + \sum_{\ell = 0}^{k-1} \distance_{\rho_1}(i_\ell, i'_\ell)\\
    &= \distance_{\rho_1}(\mathcal{S}_1) = 0. \nonumber
\end{align}
This leads to a contradiction.

Let us now assume that there exists at least one action $\nilTimerKilled$ in $\mathcal{S}_1$. Consider any $\ell$, $0 \leq \ell < k$, such that one action among $i_\ell, i'_\ell$ is equal to $\nilTimerKilled$. Necessarily, $i_{\ell} = \nilTimerKilled$ and $i'_\ell$ occurs before $i_\ell$ in $\rho_1$, that is as $\distance_{\rho_1}(i_\ell,i'_\ell) < 0$. Indeed, given the two races $B_{\ell-1 \bmod k} \prec B_\ell \prec B_{\ell+1 \bmod k}$, $\nilTimerKilled$ participates in the first race and appears at the end of $B_\ell$. It follows that 
\begin{eqnarray*}
  \distance_{\rho_1}(i_\ell,i'_\ell) < \distance_{\rho_2}(j_\ell,j'_\ell) < 0.
\end{eqnarray*}
Therefore, we get the same inequalities as in~\eqref{eq:contradiction}, leading again to a contradiction.
This completes the proof.
\qed\end{proof}

\section{Proof of \Cref{lem:encoding}}\label{proof:lem:encoding}

\lemEncoding*

Before giving the proof, some definitions
are in order: Sets of finite words (word structures, to be precise) over an
alphabet $\Sigma$ can be defined by sentences in MSO with the signature
$(<,{\{Q_a\}}_{a \in \Sigma})$. Intuitively, we interpret the formula over the
word $w \in \Sigma^*$ with variables being positions that take values in
$\mathbb{N}$, that can be ordered with $<$, and the predicates $Q_a(p)$
indicating whether the $p$-th symbol of the word (structure) is $a$. The
formulas also use variables $P$ being sets of positions, and $P(p)$ meaning
that $p$ is a position belonging to $P$.
We recall that a formula is in \emph{prenex normal form} if it can be written as \(Q_1 v_1 Q_2 v_2 \dotso Q_n v_n F\) with \(F\) a formula without
quantifiers, \(Q_i\) a quantifier and \(v_i\) a variable for all \(1 \leq i
\leq n\). We suppose the reader is familiar with the rules to put a formula into a prenex normal form. By \emph{quantifier alternations}, we mean alternating blocks of existential or universal quantifiers, respectively denoted by $\exists^*$ and $\forall^*$.

Moreover, recall that the modifications applied on the region automaton imply
the following property (see \Cref{subsec:MSO}).
Given a timed run \(\rho\) of an \AwT, by \Cref{lem:region-auto}, there exists
an equivalent run \(\rho'\) in the region automaton such that any $x$-block
$(i_{k_1} \ldots i_{k_m},\timerFate)$ of $\rho$ is translated into the sequence
of symbols $(i'_{k_1},\ldots,i'_{k_m},\timerFate')$ in $\rho'$ with an optional
symbol $\timerFate'$ such that:
\begin{itemize}
    \item $i'_{k_\ell} = (i_{k_\ell},x)$, for $1 \leq \ell < m$,
    \item $i'_{k_m} = (i_{k_m},\bot)$ if $\timerFate = \bot$, and $(i_{k_m},x)$ otherwise,
    \item $\timerFate' = \killSymbol{x}$ if $\timerFate = \nilTimerKilled$, and $\timerFate'$ does not exist otherwise.
\end{itemize}
In the sequel, we again call $\regionAutomaton$ the modified region automaton. 

We are going to describe a formula \(\nowigglingMSO\) such that a word
labeling a run $\rho$ of $\regionAutomaton$ satisfies $\nowigglingMSO$ iff
\(\rho\) is a \maximal run that cannot be wiggled. To define \(\nowigglingMSO\), we
use~\Cref{prop:wigglable_run}. Recall that it characterizes an unwigglable run
$\rho$ by a cyclic block graph $\blockGraph[\rho]$. We also focus on the
particular cycle of $\blockGraph[\rho]$ as described in~\Cref{cor:cycle}. Step
by step, we create MSO formulas expressing the following statements about a
run $\rho$ of $\regionAutomaton$:
\begin{enumerate}
  \item Two symbols belong to the same block (see the above property about the translation of $x$-blocks in $\regionAutomaton$ and the particular case of zero-valued timers that are discarded).
  \item Two blocks participate in a race. Rather, we express that two symbols, one in each block, participate in a race.
  \item The run is a \maximal run that cannot be wiggled. Rather, we express that there exists a cycle in $\blockGraph[\rho]$ whose form is as in \Cref{cor:cycle}.
\end{enumerate}

\subsubsection{Some useful predicates.}
We define four predicates to help us write the MSO formulas. The formula \(\firstMSO(p, P)\) expresses that a position \(p\) is the first element of a set \(P\), while \(\lastMSO(p, P)\) states that \(p\) is the last element of \(P\). Finally, \(\nextMSO(p, P, q)\) expresses that \(q\) is the successor of \(p\) in $P$ with regards to \(<\). More formally,
\begin{gather} 
  \firstMSO(p, P) \coloneqq P(p) \land \forall q \colon q < p \impliesMSO \neg P(q)
  \label{eq:MSO:first}\\
  \lastMSO(p, P) \coloneqq P(p) \land \forall q \colon q > p \impliesMSO \neg P(q)
  \label{eq:MSO:last}\\
  \nextMSO(p, P, q) \coloneqq p < q \land P(p) \land P(q) \land \forall r \colon
  p < r < q \impliesMSO \neg P(r).
  \label{eq:MSO:next}
\end{gather}
The last useful predicate, $\partition(P,P_1,P_2)$, states that there exist sets of positions \(P,
P_1, P_2\) such that \(P = P_1 \uplus P_2\), the first position of \(P\) is in \(P_1\), the last one is in \(P_2\), and the positions of \(P\) alternate between \(P_1\) and \(P_2\):
\begin{gather}  \label{eq:MSO:partition}
  \begin{split}
    \partition(P,P_1,P_2) &\coloneqq \forall r \colon P(r) \iffMSO (P_1(r) \vee P_2(r)) \\
      &\land \exists p, q \colon \firstMSO(p, P) \land \lastMSO(q, P)
        \land P_1(p) \land P_2(q)\\
      &\land \forall r \colon P_1(r) \iffMSO \neg P_2(r)\\
      &\land \forall r, s \colon \nextMSO(r, P, s) \impliesMSO (P_1(r) \iffMSO P_2(s)).
  \end{split}
\end{gather}

\subsubsection{Two symbols belong to the same block.}
We give here an MSO formula expressing that two positions \(p < q\) are labeled by symbols belonging to the same $x$-block. Thus, the formula \(\blockMSO_x(p, q, P)\), with \(P\) a set of positions labeled by consecutive symbols of an $x$-block, states that \(P\) must respect the following constraints:
\begin{itemize}
  \item We have \(p < q\), with \(p, q \in P\).
  \item The position \(p\) is labeled by either \((i, x) \in \Sigma\) (meaning we start an $x$-block~$B$), or \((\timeout{x}, x) \in \Sigma\) (meaning we are in the block $B$),
  \item The input at position \(q\) is either \((\timeout{x}, x)\) (we are in the block $B$), or \((\timeout{x}, \bot)\) (we are at the end of $B$ and we do not restart \(x\)), or \(\killSymbol{x}\) (we finish $B$ by discarding its timer while its value is zero),
  \item Every other position $r \in P$ such that $p<r<q$ is labeled by \((\timeout{x}, x)\) (we restart \(x\) to keep the block active),
  \item There is no position \(r \notin P\) between \(p\) and \(q\) such that \(r\) is labeled by some symbol of $\Sigma$ among \((\timeout{x}, \cdot), (\cdot, x)\), or \(\killSymbol{x}\) (the \(\cdot\) indicates \enquote{any value}). That is, any intermediate position cannot affect \(x\).
\end{itemize}
Formally, we have:
\begin{gather}
  \mathrm{Affects}_x(r) \coloneqq Q_{(\timeout{x}, x)}(r) \lor Q_{(\timeout{x}, \bot)}(r) \lor \bigvee_{i \in I} Q_{(i, x)}(r) \lor Q_{\killSymbol{x}}(r) \label{eq:MSO:Affects}\\
  \label{eq:MSO:block}
  \begin{split}
    \blockMSO_x(p, q, P) &\coloneqq p < q
    \land P(p) \land P(q)\\
    &\land \big(\bigvee_{i \in I} Q_{(i, x)}(p) \lor Q_{(\timeout{x}, x)}(p)\big)\\
    &\land \big(Q_{(\timeout{x}, x)}(q) \lor Q_{(\timeout{x}, \bot)}(q) \lor
      Q_{\killSymbol{x}}(q)\big)\\
    &\land \forall r \colon \big(p < r < q \land P(r)\big) \impliesMSO 
      Q_{(\timeout{x}, x)}(r)\\
    &\land \forall r \colon \big(p < r < q \land \neg P(r)\big) \impliesMSO
      \neg \mathrm{Affects}_x(r).
  \end{split}
\end{gather}

\subsubsection{Two symbols participate in a race.}
The formula \(\raceMSO(p, q)\) states that two positions \(p < q\) are labeled by symbols that participate in a race, that is, there is no position labeled by $\tau$ between $p$ and $q$. 
\begin{equation}\label{eq:MSO:race}
  \raceMSO(p, q) \coloneqq p < q \land \neg\big(\exists r \colon p < r < q
    \land Q_{\delaySymbol}(r)\big).
\end{equation}

\subsubsection{The run is unwigglable.}
Finally, we give a formula \(\nowigglingMSO\) that expresses that a word is the label of a \maximal run \(\rho\) that cannot be wiggled, i.e., that highlights a cycle (as in~\Cref{cor:cycle}) in the block graph of $\rho$. The idea of that formula is as follows. There are positions \(p_1 < q_1 < p_2 < q_2 < \dotsb < p_m < q_m\) such that each pair \(p_k, q_k\) participate in a race. Moreover, for any \(q_k\) belonging to a block \(B_k\), there must exist a \(p_\ell\) that also belongs to \(B_k\). Notice that \(p_\ell\) is not necessarily after \(q_k\). See \Cref{fig:races} for an illustration of that scenario with $m = 5$.

The formula \(\nowigglingMSO\) states that there exist sets of positions \(P,P_1, P_2\) such that:
\begin{itemize}
  \item We have $P_1 \uplus P_2$ forming a partition of $P$ as described previously,
  \item For any \(p \in P_1\) and \(q \in P_2\) such that \(q\) is the next
    element after \(p\) in \(P\), we have \(\raceMSO(p, q)\).
  \item For any \(q \in P_2\), there must exist a \(p\) in \(P_1\) such that
    \(p\) and \(q\) belong to the same block.
    That is, there must exist a timer \(x\) and a set \(P'\) such that
    \(\blockMSO_x(p, q, P')\) if \(p < q\) or \(\blockMSO_x(q, p, P')\) if \(p > q\).
\end{itemize}
Finally, in order to describe a \maximal run $\rho$, the first and last positions of the word must be labeled with $\tau$ (i.e., a non-zero delay). These positions do not belong to $P$.
\begin{gather}
  \mathrm{InBlock}(p, q, P') \coloneqq
    \big(p < q \land \bigvee_{x} \blockMSO_x(p, q, P')\big)
    \lor \big(p > q \land \bigvee_{x} \blockMSO_x(q, p, P')\big) \label{eq:MSO:InBlock}\\
  \label{eq:MSO:NoWiggling}
  \begin{split}
    \nowigglingMSO &\coloneqq \exists P, P_1, P_2 \colon \Big(\partition(P,P_1,P_2)\\
      &\land \forall p, q \colon \big(P_1(p) \land \nextMSO(p, P, q)\big) \impliesMSO
        \raceMSO(p, q)\\
      &\land \forall q \colon P_2(q) \impliesMSO \big(\exists p, P' \colon P_1(p)
        \land \mathrm{InBlock}(p, q, P')\big)\Big) \\
        &\land Q_\tau(1) \land \big(\exists r \colon Q_\tau(r) \land (\forall r' \colon r' < r)\big).
  \end{split}
\end{gather}

\subsubsection{Correctness.} Now that we have constructed the formula $\nowigglingMSO$, let us show that it correctly encodes that the word labeling a run $\rho$ in $\regionAutomaton$ satisfies $\nowigglingMSO$ iff $\rho$ is a \maximal run that is not wigglable.

  \(\Leftarrow\)
  Assume $\rho$ in $\regionAutomaton$ is a \maximal run that is not wigglable. Let $w \in \Sigma^*$ be its labeling. By \Cref{prop:wigglable_run}, we know that the block graph \(\blockGraph[\rho]\) is cyclic. Moreover, by~\Cref{cor:cycle}, there exists a cycle whose blocks satisfy the following properties: exactly two blocks participate in any race of the cycle, any block participates in exactly two races, and any block has a size at least equal to two\footnote{Recall the way blocks in $\A$ are translated into blocks in $\regionAutomaton$.}. We consider this particular cycle \((B_0, \ldots, B_{m-1}, B_0)\).

  From the races in which the blocks $B_k$ participate, we define the sets \(P_1, P_2\) of positions, and, thus, \(P = P_1 \uplus P_2\) as follows. For every \(B_{k} \prec B_{k+1 \bmod m}\) in the cycle, consider $a \in B_k$ and $b \in B_{k+1 \bmod m}$ that are the two symbols of $w$ participating in a race. We add the position of $a$ in \(P_1\) and the position of $b$ in \(P_2\) (see~\Cref{fig:races} to get intuition). Thus, the positions in \(P\) alternate between \(P_1\) and \(P_2\), the first element of \(P\) is in \(P_1\), and the last is in \(P_2\). Moreover, for any \(p \in P_1\) and \(q \in P_2\) such that \(q\) is the successor of \(p\) in \(P\), it holds that \(\raceMSO(p, q)\). That is, the second line of formula~\eqref{eq:MSO:NoWiggling} about races is satisfied by $w$.

  We have to show that the third line of~\eqref{eq:MSO:NoWiggling} is also satisfied by $w$, that is, for any position \(q\) in \(P_2\), there exists a position \(p\) in \(P_1\) such that \(p\) and \(q\) belong to the same block. Let \(q \in P_2\). By construction, \(q\) belongs to some block \(B_k\) of the cycle. By~\Cref{cor:cycle}, $B_k$ participate in exactly two races of the cycle, one as described above with \(\raceMSO(p, q)\), and another one with \(\raceMSO(p', q')\) for some other positions $p' \in P_1$ and $q' \in P_2$. Necessarily, $p'$ is a position of a symbol in $B_k$ (and not $q'$ by choice of the cycle), such that either $p' < q$ or $p'> q$. Thus, the third line of~\eqref{eq:MSO:NoWiggling} is satisfied by $w$.

  Finally, since $\rho$ is \maximal, it must be that its first and last delays are positive, i.e., the corresponding positions are labeled by $\tau$. Hence the last line of~\eqref{eq:MSO:NoWiggling} is satisfied by $w$. We conclude that $w$ satisfies all conjuncts of~\eqref{eq:MSO:NoWiggling} and then also the formula \(\nowigglingMSO\). 
  
  \(\Rightarrow\)
  Assume now that the label $w$ of a run $\rho$ in $\regionAutomaton$ satisfies formula $\nowigglingMSO$. 
  Since the last line of~\eqref{eq:MSO:NoWiggling} forces the first and last symbols of $w$ to be \(\tau\), the formula describes a \maximal run of \(\regionAutomaton\).

  Let $P, P_1$, and $P_2$ be the sets that satisfy the formula $\nowigglingMSO$. Let \(P_1 = \{p_1, \dotsc, p_m\}\) and \(P_2 = \{q_1, \dotsc, q_m\}\) be such that \(\nextMSO(p_k, P, q_k)\) is satisfied for all \(k\). Then, by~\eqref{eq:MSO:race}, it holds that \(\raceMSO(p_k, q_k)\) is also satisfied, i.e., the symbols of $w$ labeling the positions $p_k, q_k$ participate in a race. The third line of~\eqref{eq:MSO:NoWiggling} implies that there are at most \(m\) blocks involved in these races. Notice that there can be less than \(m\) blocks, as for \(q, q'\) in \(P_2\) with $q \neq q'$, we could have the same \(p \in P_1\) that makes sub-formula $\mathrm{InBlock}$ satisfied in~\eqref{eq:MSO:NoWiggling}. 
  
  From formula $\nowigglingMSO$, we are going to construct a part of the block graph $\blockGraph[\rho]$ of \(\rho\) that is cyclic. We proceed inductively as follows:
  \begin{itemize}
    \item Take an arbitrary position \(q_{k_0} \in P_2\). There exists \(p_{k_1} \in P_1\) such that \(\blockMSO_{x_1}(q_{k_0}, p_{k_1}, P'_1)\) or \(\blockMSO_{x_1}(p_{k_1}, q_{k_0}, P'_1)\) is satisfied. Call $B_1$ the related $x_1$-block.
    \item Let \(q_{k_1} \in P_2\). Then, there exists \(p_{k_2} \in P_1\) such that
      \(\blockMSO_{x_2}(q_{k_1}, p_{k_2}, P'_2)\) or
      \(\blockMSO_{x_2}(p_{k_2}, q_{k_1}, P'_2)\) is satisfied.
      For the related $x_2$-block $B_2$, we have \(B_1 \prec B_2\) as
      $p_{k_1}, q_{k_1}$ participate in a race, and \(p_{k_1} \in B_1, q_{k_1} \in B_2\).
    \item Let \(q_{k_2} \in P_2\). Then, there exists \(p_{k_3} \in P_1\) such that
      \(\blockMSO_{x_3}(q_{k_2}, p_{k_3}, P'_3)\) or
      \(\blockMSO_{x_3}(p_{k_3}, q_{k_2}, P'_3)\) is satisfied.
      For the related $x_3$-block $B_3$, we have \(B_2 \prec B_3\).
    \item We repeat this process until we obtain a cycle. This situation necessarily
      arises as the number of blocks is bounded by \(m\).
  \end{itemize}
  This shows that $\blockGraph[\rho]$ is cyclic.

\subsubsection{Prenex form of $\nowigglingMSO$.}

Formula \(\blockMSO_x(p, q, B)\), see~\eqref{eq:MSO:block}, can be easily rewritten in prenex normal form that starts with a block \(\forall^*\) of universal quantifiers.
Similarly $\raceMSO(p,q)$, see~\eqref{eq:MSO:race}, requires a single universal quantifier. Let us consider the formula \(\nowigglingMSO\) putting aside the quantifiers $\exists P,P_1,P_2$, see~\eqref{eq:MSO:NoWiggling}. Notice that formulas~\eqref{eq:MSO:first},~\eqref{eq:MSO:last}, and~\eqref{eq:MSO:next} all use a single universal quantifier. The first conjunct of~\eqref{eq:MSO:NoWiggling} uses $\partition(P,P_1,P_2)$, see~\eqref{eq:MSO:partition}, that can be rewritten with three blocks $\forall^*\exists^*\forall^*$. The second (resp.\ last) conjunct can be rewritten with two blocks $\forall^* \exists^*$ (resp. $\exists^*\forall^*$), and the third one with three blocks $\forall^*\exists^*\forall^*$. Hence, we obtain that the quantifiers of the prenex normal form of \(\nowigglingMSO\) are \(\exists^* \forall^* \exists^* \forall^*\), that is, with three quantifier alternations, as expected.

Finally, by carefully examining the formulas, we notice that most of them have constant size except $\mathrm{Affects}_x(r)$ and $\blockMSO_x(p, q, P)$ whose sizes are linear in $|I|$, and $\mathrm{InBlock}(p, q, P')$ and $\nowigglingMSO$ whose sizes are linear in $|I|$ and $|X|$.

\section{Proof of \Cref{thm:robust}}

\thmRobust*

We begin by proving the upper bound.

\subsection{Upper bound}
We make use of the B\"uchi-Elgot-Trakhtenbrot theorem: A language is regular
iff it can be defined as the set of all words satisfied by an MSO
formula (with effective translations,
see~\cite{gradel2003automata,DBLP:reference/hfl/Thomas97}). First, from the
formula $\nowigglingMSO$ of \Cref{lem:encoding}, we can construct a
finite-state automaton $\N$ whose language is the set of all words
satisfying $\nowigglingMSO$. Due to the reduction from non-deterministic to
deterministic automaton, each quantifier alternation induces an exponential
blowup in the automaton construction. Hence, the size of $\N$ is
triple-exponential as $\nowigglingMSO$ has three quantifier alternations.
Second, we compute the intersection of $\N$ with \(\regionAutomaton\) ---
itself exponential in size.  This can be done in polynomial time in the
sizes of both automata, i.e., in \THREEEXP{}. Finally, the language of the
resulting automaton is empty iff there is no \maximal run of $\A$ that
cannot be wiggled, and this can be checked in polynomial time with respect
to the size of the (triple-exponential) automaton.

\subsection{Lower bound}

In this section, we prove that deciding whether all \maximal timed runs of a given \AwT are wigglable is a \PSPACE-hard problem. The idea of the proof is to leverage the \PSPACE-hardness proof for the reachability problem, see \Cref{thm:reachability}. We use the same notations as in the proof provided for the latter result in \Cref{sec:reachability:lower}.

Let $\M$ be an LBTM and $w$ be an input word. Let $\A_{\M, w}$ be the \AwT constructed from $\M$ and $w$ in the proof of \Cref{thm:reachability}, such that the state $r_\done$ is reachable in the \AwT iff the LBTM accepts $w$. We are going to slightly modify the LBTM and the constructed \AwT such that any \maximal timed run \(\rho\) of $\A_{\M, w}$ has an acyclic block graph \(\blockGraph[\rho]\). Thanks to \Cref{prop:wigglable_run}, this is equivalent to stating that \(\rho\) can be wiggled. Then, we give a widget that extends any \maximal timed run \(\rho\) reaching $r_{done}$ into one that is unwigglable. Hence, the given LBTM accepts $w$ if the constructed \AwT has some unwigglable \maximal timed run.

First, we modify the LBTM \(\M\) and word \(w = w_1 \dotsc w_n\) to obtain a new LBTM \(\M'\) that accepts \(w\) iff it does so while maintaining the invariant that no two cells of the tape of \(\M'\) hold the same symbol. Let \(\Sigma\) be the alphabet of \(\M\). We create a new word \(w' = (w_1, 1) (w_2, 2) \dotso (w_n, n)\), and a new LBTM \(\M'\) over the alphabet \(\Sigma' = \Sigma \times \{1, \dotsc, n\}\) such that \(\M'\) simulates \(\M\) by discarding the second component of each symbol of $\Sigma'$, except that whenever \(\M\) writes the symbol \(a\) at the position \(i\) on the tape, \(\M'\) writes the symbol \((a, i)\). This requires to store the current \(i\) in the state of $\M'$, inducing a polynomial blowup in \(n\) for the number of states of $\M'$. Thanks to the second component, we indeed have that every cell contains a symbol that is different from any other cell.

\begin{figure}[t]
  \centering
  \input{figures/blocks/pspace-hard/timerY.tex}
  \caption{Visualization of the forced buffer using timer \(y\).}%
  \label{fig:buffer}
\end{figure}
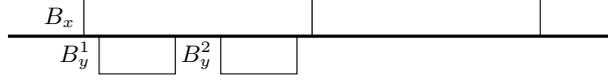

Second, we modify the construction of the \AwT \(\A_{\M', w'}\) as follows. We add a new timer \(y\) which we use to force a block acting as a buffer before and after each action implying the timer \(x\) when the value of $\clock$ is equal to zero, as illustrated in \Cref{fig:buffer}. Therefore, any input-action must take place between these two buffers. In order to have enough room for these buffers, we multiply by three all the values at which the timers \(x, x_1, \dotsc, x_n\) are updated (in particular $x$ is always (re)started with value $3$ instead of $1$). For the initialization of the timers \(x, x_1, \dotsc, x_n\) (at the start of phase~$1$), we add new states $\langle r_1,y_k \rangle$ and $\langle q_0, 1, 0, 0, y_k \rangle$, for $k = 1,2$, and modify the transitions to force \(y\) to start and time out, then we force the initialization part and, finally, force \(y\) to start and time out again:
\begin{gather*}
  r_0
  \xrightarrow[(x, 3)]{go} \langle r_1, y_1 \rangle
  \xrightarrow[(y, 1)]{go} \langle r_1, y_2 \rangle
  \xrightarrow[\bot]{\timeout{y}} r_1 \ldots\\
  \ldots r_n
  \xrightarrow[(x_n, 6j)]{go} \langle q_0, 1, 0, 0, y_1 \rangle
  \xrightarrow[(y, 1)]{go} \langle q_0, 1, 0, 0, y_2 \rangle
  \xrightarrow[\bot]{\timeout{y}} \langle q_0, 1, 0, 0 \rangle
\end{gather*}
with \(6j\) the value at which \(x_n\) must be set (see the proof of
\Cref{thm:reachability}). We define $\chi(\langle r_1,y_1 \rangle) = \chi(r_1) = \{x\}$, $\chi(\langle r_1,y_2 \rangle) = \{x,y\}$, $\chi(\langle q_0, 1, 0, 0, y_1 \rangle) = \chi(\langle q_0, 1, 0, 0 \rangle) = \{x,x_1,\ldots,x_n\}$, and $\chi(\langle q_0, 1, 0, 0, y_2 \rangle) = \{x,x_1,\ldots,x_n,y\}$. In states where $\clock = 0$ in the other phases, we do likewise in the following way. For each state \(\langle q, i, \tsymbol, 0 \rangle\) of the \AwT with \(\tsymbol > 0\), we add new states \(\langle q, i, \tsymbol, 0, y_k \rangle\) and $\langle q', j, 0, 0, y_k \rangle$, for $k = 1,2$, and modify the transitions to force \(y\) to start and time out, then allow a \(t\)-transition\footnote{We use the notations of the proof of \Cref{thm:reachability}.}, and, finally, force \(y\) to start and time out again:
  \begin{align*}
    \langle q, i, \tsymbol, 0, y_1 \rangle
    &\xrightarrow[(y, 1)]{\go} \langle q, i, \tsymbol, 0, y_2 \rangle
    \xrightarrow[\bot]{\timeout{y}} \langle q, i, \tsymbol, 0 \rangle\\
    &\xrightarrow[(x_i, 6j)]{t} \langle q', j, 0, 0, y_1 \rangle
    \xrightarrow[(y, 1)]{\go} \langle q', j, 0, 0, y_2 \rangle
    \xrightarrow[\bot]{\timeout{y}} \langle q', j, 0, 0 \rangle.
  \end{align*}
Any \(\timeout{x}\)-transition leading to \(\langle q, i, \tsymbol, 0 \rangle\) instead goes to \(\langle q, i, \tsymbol, 0, y_1 \rangle\). We then define $\chi(\langle q, i, \tsymbol, 0, y_1 \rangle) = \chi(\langle q', j, 0, 0, y_1 \rangle) = \{x,x_1,\ldots,x_n\}$ and  $\chi(\langle q, i, \tsymbol, 0, y_2 \rangle) = \chi(\langle q', j, 0, 0, y_2 \rangle) = \{x,x_1,\ldots,x_n,y\}$. Any other transition leads to \(r_\sink\).

Third, let \(\A'\) be the \AwT obtained with this modified construction. We now argue that any \maximal timed run \(\rho \in \mtruns{\A'}\) can be wiggled. We do so by proving that the block graph of \(\rho\) is acyclic. In the sequel, we say that two block $B, B'$ are \emph{incomparable} if neither $B \prec B'$ nor $B' \prec B$. 

\begin{enumerate} 
\item Suppose first that $\rho$ does not contain the state $r_{\sink}$.
\begin{itemize}
\item Recall that \(x\) is never discarded and every \(\timeout{x}\)-transition restarts it. Thus, there exists a single \(x\)-block; we call it \(B_x\). By construction, none of the \(\timeout{y}\)-transitions update \(y\). Therefore, we have strictly more than one \(y\)-block. Call them \(B_y^1, B_y^2, \dotsc, B_y^{m_y}\) in the order they are seen alongside \(\rho\). For any odd \(i\), it may be that \(B_x \prec B_y^i\) (if the input-action starting \(y\) occurs at the same time \(x\) times out), \(B_y^i \prec B_y^{i+1}\) (if the sum of the delays between the timeout of \(y\) in $B_y^i$ and the \(\go\)-transition starting $B_y^{i+1}$ is zero), and \(B_y^{i+1} \prec B_x\). However, it is impossible to have \(B_x \prec B_y^i \prec B_y^{i+1} \prec B_x\), as $x$ is (re)started with value $3$ and there are exactly two units of time if \(B_y^i \prec B_y^{i+1}\) (since \(y\) is always started with value $1$). That is, we do not have a cycle. Moreover, since the \(y\)-blocks only appear when the \(\clock\) component of the current state is zero, \(B_y^j\) and \(B_y^{j+1}\) are incomparable for every even \(j\).
\item Let us now focus on the timers \(x_i\). Since a \(t\)-action can discard \(x_i\), we may have multiple \(x_i\)-blocks, say \(B_{x_i}^1, \dotsc, B_{x_i}^{m_{x_i}}\) (again, in the order they are seen in \(\rho\)). Since \(x_i\) is updated such that it cannot time out while the \(\clock\) component of the current state is zero (i.e., the timer fate of the corresponding block is \(\nonNilTimerKilled\)) and only one \(t\)-transition can occur per phase of \(\M'\), all of the \(x_i\)-blocks are pairwise incomparable. Moreover, thanks to the \(y\)-buffers, \(B_x\) and any block \(B_{x_i}^j\) are incomparable. As stated before, it may happen that a \(t\)-transition of some block \(B_{x_i}^j\) occurs concurrently with an action of a block \(B_y^\ell\) (resp. \(B_y^{\ell+1}\)) with \(\ell\) an odd number. By construction, \(B_{y}^\ell \prec B_{x_i}^j\) (resp. \(B_{x_i}^j \prec B_{y}^{\ell+1}\)). Note that it is possible that \(B_y^\ell \prec B_{x_i}^j \prec B_y^{\ell+1}\) if all blocks participate in the same race.
\item We now consider two different timers \(x_i\) and \(x_k\). Since the LBTM \(\M'\) is such that no two cells contain the same value, it must be that \(x_i\) and \(x_k\) are always updated to time out in states containing different \(\clock\) component in \(\A'\) (as, otherwise, this would imply that the two cells of \(\M'\) contain identical symbols). That is, it is not possible for two timers to time out concurrently. However, during the initialization, it may be that the actions starting \(B_{x_1}^1, B_{x_2}^1, \dotsc, B_{x_n}^1\) occur at the same time. We thus have \(B_{x_1}^1 \prec B_{x_2}^1 \prec \dotsb \prec B_{x_n}^1\) and the blocks \(B_{x_i}^j\) and \(B_{x_k}^\ell\) are incomparable with \(j, \ell > 1\).
\end{itemize}
Using all these facts over the races and \(\prec\), we deduce that \(\blockGraph[\rho]\) is acyclic, i.e., \(\rho\) can be wiggled. 

\item Suppose now that $\rho$ contains the state $r_\sink$.
Recall that \(\chi(r_\sink) = \emptyset\), meaning that the update of every
transition ending in $r_\sink$ is \(\bot\), every timer is stopped, and every new
block started after reaching $r_\sink$ contains exactly one action.
We thus focus on the prefix of the run up to the transition leading to $r_\sink$.
Call this last transition $t^*$ with action $i^*$.
By construction, it must be this prefix is a run that satisfies the constraints
explained above (i.e., there is no cycle in the block graph induced by the prefix
of the run).
Let us show that adding $t^*$ after the prefix does not induce a cycle in
the block graph.

Suppose first that $i^*$ is an input. As the update of $t^*$ is $\bot$, it must be that $i^*$ is the only action of its block, and this block cannot appear in a cycle.
Suppose now that $i^*$ is a timeout-action. 
\begin{itemize}
\item Let us study the initialization or a simulation step (i.e., the start of a phase).
Recall that no timer \(x_i\) can time out, i.e., we only have to consider the timers \(x\) and \(y\). In this case, we must have $i^* = \timeout{x}$. Indeed, this happens when a block $B_y^\ell$ is started too late, in a way that $\timeout{x}$ occurs before $y$ times out, or when such a block is not started at all. The only hope to have a cycle is to adapt to the current situation the discussion we made above about \(B_x \prec B_y^\ell \prec B_y^{\ell+1} \prec B_x\) with $\ell$ odd. Here, as $i^* = \timeout{x}$ stops the timer $y$ and $x$ (resp. $y$) is (re)started with value 3 (resp.\ value $1$), we get $B_x \prec B_y^\ell$ and $B_x \prec B_y^{\ell+1}$. We thus have no cycle. 
\item Otherwise, we consider a state in which the \(\clock\) value is not zero (i.e., this is not the start of a phase). By construction, a \(\timeout{x}\)-, or any \(\timeout{x_i}\)-transition cannot lead to \(r_\sink\), and \(y\) is never started.
Thus, $i^*$ cannot be a timeout-action in this case.
\end{itemize}

Hence, we covered every case and never obtained a cycle, i.e., $\rho$ is wigglable.
\end {enumerate}
We have thus proved that each timed run \(\rho \in \mtruns{\A'}\) can be wiggled.

Finally, we add a widget that forces an unwigglable run after \(r_\done\). To do so, we add new timers \(z, z'\) and states \(s_1\) to \(s_4\), and define the following transitions
  \begin{equation*}
    r_\done
      \xrightarrow[(z, 1)]{\go} s_1
      \xrightarrow[(z', 1)]{\go} s_2
      \xrightarrow[\bot]{\timeout{z'}} s_3
      \xrightarrow[\bot]{\timeout{z}} s_4.
  \end{equation*}
  We define $\chi(s_1) = \chi(s_3) = \{z\}$, $\chi(s_2) = \{z,z'\}$, and $\chi(s_4) = \emptyset$.
  Given a \maximal timed run ending in $r_\done$, the only ways to extend it into a \maximal timed run reaching \(s_4\) are by adding the following sequence of transitions, with any $d > 0$:
  \begin{equation}\label{eq:widget}
    \begin{split}
        &(r_\done, \emptyset)
          \xrightarrow[(z, 1)]{\go}             (s_1, z = 1)
          \xrightarrow{0}                       (s_1, z = 1)
          \xrightarrow[(z', 1)]{\go}            (s_2, z = 1, z' = 1) \\
          &\xrightarrow{1}                      (s_2, z = 0, z' = 0)
          \!\xrightarrow[\bot]{\timeout{z'}}\!  (s_3, z = 0)
          \xrightarrow{0}                       (s_3, z = 0)
          \xrightarrow[\bot]{\timeout{z}}       (s_4, \emptyset)
          \xrightarrow{d}                       (s_4, \emptyset).
    \end{split}
  \end{equation}
The resulting \maximal timed run is not wigglable, as two blocks $B_z$ and $B_{z'}$ have been added such that \(B_{z} \prec B_{z'} \prec B_{z}\).

To conclude, it remains to prove that the given \(\M'\) accepts the given \(w'\) iff there exists an unwigglable \maximal timed run in the \AwT $\A'$ extended with the widget. First, suppose that $\M'$ accepts $w'$. Then by the proof of \Cref{thm:reachability}, we know that there exists a timed run $\rho \in \truns{\A'}$ reaching $r_\done$. As both $r_0$ and $r_\done$ do not have any active timer, the first and last delays of $\rho$ can be made positive, thus making $\rho$ \maximal. We then extend $\rho$ with (\ref{eq:widget}) and obtain an unwigglable run that is still \maximal. Second, suppose that there exists a \maximal timed run $\rho$ that cannot be wiggled. We proved above that if $\rho$ ends in some state of $\A'$, then $\rho$ is wigglable. Therefore, by construction of the widget, $\rho$ has to end with $s_4$, meaning that a prefix of $\rho$ reaches $r_\done$. It follows that $w'$ is accepted by $\M'$. Thus, deciding whether all \maximal timed runs of an \AwT can be wiggled is a \PSPACE-hard problem.

%% file: figures/blocks/pspace-hard/init.tex
\def\numberBlocksStart{3}

\begin{tikzpicture}
    \coordinate (startBx)   at (0.5, 0);
    \coordinate (toBx)      at (5, 0);
    \coordinate (end)       at (5.5, 0);

    \draw[block]
        let
            \p{s} = (startBx),
            \p{t} = (toBx)
        in
            (\p{s})
                -- (\x{s}, 0.5)
                --  node [above] {\(B_x\)}
                    node [below] {0}
                    (\x{t}, 0.5)
                -- (\p{t})
            (\x{t}, 0.5)
                -- ($(end) + (0, 0.5)$)
    ;
    \draw [loosely dotted, thick, dash phase=20pt]
        ($(end) + (0, 0.5)$)
            -- ($(end) + (0.5, 0.5)$)
    ;

    \foreach \i [
        evaluate=\i using \i-1,
        evaluate=\i using 4-\i as \j,
        count=\counter from 1,
    ]
    in {1, ..., \numberBlocksStart} {
        \draw[block]
            let
                \p{s} = ($(startBx) + (0.2, 0) + (0.8 * \i, 0)$),
                \p{end} = (end),
                \p{t} = ($(\x{end}, -0.5) + (0, -0.3 * \j)$),
            in
                (\p{s})
                    -- (\x{s}, \y{t})
                    -- (\p{t})
        ;
        \draw
            let
                \p{s} = ($(startBx) + (0.2 + 0.8 * \i, -0.5)$),
            in
                node [left] at (\p{s}) {\(B_{x_{\counter}}\)}
        ;
        \draw [loosely dotted, thick, dash phase=20pt]
            let
                \p{end} = (end),
                \p{t} = ($(\x{end}, -0.5) + (0, -0.3 * \j)$),
            in
                (\p{t}) -- ($(end) + (0.5, \y{t})$)
        ;
    }

    \node at ($(startBx) + (0.2 + 0.8 * \numberBlocksStart, -0.5)$) {\(\dotso\)};

    \draw[block]
        let
            \p{s} = ($(startBx) + (4, 0)$),
            \p{end} = (end),
            \p{t} = ($(\x{end}, -0.8)$),
        in
            (\p{s})
                -- (\x{s}, \y{t})
                -- (\p{t})
    ;
    \draw
        let
            \p{s} = ($(startBx) + (4, -0.5)$),
        in
            node [left] at (\p{s}) {\(B_{x_{n}}\)}
    ;
    \draw [loosely dotted, thick, dash phase=20pt]
        let
            \p{end} = (end),
            \p{t} = ($(\x{end}, -0.8)$),
        in
            (\p{t}) -- ($(end) + (0.5, \y{t})$)
    ;

    \draw [timeline]
        (0, 0) -- (end)
    ;
    \draw [loosely dotted, thick, dash phase=20pt]
        (end) -- ($(end) + (0.5, 0)$)
    ;
\end{tikzpicture}

%% file: figures/blocks/pspace-hard/run.tex
\def\numberPhases{2}
\def\numberBlocksPerPhase{6}
\def\blockLength{0.6}
\def\blockHeight{0.5}
\def\offsetFirstPhase{0.2}

\def\phaseLength{\numberBlocksPerPhase*\blockLength}
\def\endFigure{\numberPhases*\phaseLength}
\begin{tikzpicture}[
    phaseLine/.style = {
        very thick,
    },
    upwardBrace/.style = {
        decorate,
        decoration = brace,
    },
    downwardBrace/.style = {
        decorate,
        decoration = {brace, mirror},
    },
]
    \coordinate (end)   at
        ($(\offsetFirstPhase + \numberPhases*\phaseLength + \phaseLength, 0)$)
    ;

    \node [above] at ($(\offsetFirstPhase + .5*\blockLength, \blockHeight)$) {\(B_x\)};
    \draw [phaseLine]
        (\offsetFirstPhase, 2*\blockHeight)
        -- (\offsetFirstPhase, -2*\blockHeight)
    ;

    \foreach \phase [
        evaluate=\phase using \phase*\phaseLength + \offsetFirstPhase as \offset,
        evaluate=\phase using (\phase+1)*\phaseLength + \offsetFirstPhase as \nextOffset,
        count=\counterPhase from 1,
    ] in {0, ..., \numberPhases} {
        \foreach \i [
            evaluate=\i using \i - 1,
            evaluate=\i using \i * \blockLength,
            evaluate=\i using \offset + \i,
            evaluate=\i using \i+\blockLength as \j,
            count=\counter from 0,
        ] in {1, ..., \numberBlocksPerPhase} {
            \draw[block]
                (\i, 0)
                    -- (\i, \blockHeight)
                    -- node [below] {\counter} (\j, \blockHeight)
                    -- (\j, 0)
            ;
        }

        \draw [phaseLine]
            (\nextOffset, 2*\blockHeight)
            -- (\nextOffset, -2*\blockHeight)
        ;

        \draw [upwardBrace, phaseLine]
            (\offset, 2*\blockHeight)
                -- node [above=5pt] {Phase \counterPhase}
                    (\nextOffset, 2*\blockHeight)
        ;
    }

    \draw [downwardBrace]
        (\offsetFirstPhase, -2*\blockHeight)
            -- node [below, align=center]  {Initialization\\of the timers}
                ($(\offsetFirstPhase + \blockLength, -2*\blockHeight)$)
    ;
    \draw [downwardBrace]
        ($(\offsetFirstPhase + \phaseLength, -2*\blockHeight)$)
            -- node [below, align=center] {Simulation\\of a step of \(\M\)}
                ($(\offsetFirstPhase + \phaseLength + \blockLength, -2*\blockHeight)$)
    ;
    \draw [downwardBrace]
        ($(\offsetFirstPhase + 2*\phaseLength, -2*\blockHeight)$)
            -- node [below, align=center] {Simulation}
                ($(\offsetFirstPhase + 2*\phaseLength + \blockLength, -2*\blockHeight)$)
    ;

    \coordinate (startB1)   at ($(\offsetFirstPhase, 0) + (.8*\blockLength, 0)$);
    \coordinate (1toB1)     at ($(startB1) + (4*\blockLength, 0)$);
    \coordinate (2toB1)     at ($(1toB1) + (\phaseLength, 0)$);
    \coordinate (diB1)      at ($(2toB1) + (1.4*\blockLength, 0)$);
    \coordinate (3toB1)     at ($(diB1) + (2*\blockLength, 0)$);
    
    \foreach \s/\t/\h/\l in {
        startB1/1toB1/-\blockHeight/B_{x_i}^1,
        1toB1/2toB1/-\blockHeight/
    } {
        \draw[block]
            let
                \p{s} = (\s),
                \p{t} = (\t)
            in
                (\p{s})
                    -- (\x{s}, \h)
                    -- node [below] {\(\l\)} (\x{t}, \h)
                    -- (\p{t})
        ;
    }

    \draw[block]
        let
            \p{s} = (2toB1),
        in
            (\p{s})
                -- (\x{s}, -\blockHeight)
    ;
    \draw[blockBeforeKilled]
        let
            \p{s} = (2toB1),
            \p{t} = (diB1),
        in
            (\x{s}, -\blockHeight)
                -- (\x{t}, -\blockHeight)
            node[nonNullTimerKilled] at (\x{t}, -\blockHeight) {}
    ;

    \foreach \s/\t/\h/\l in {
        diB1/3toB1/-\blockHeight/B_{x_i}^2,
        3toB1/end/-\blockHeight/
    } {
        \draw[block]
            let
                \p{s} = (\s),
                \p{t} = (\t)
            in
                (\p{s})
                    -- (\x{s}, \h)
                    -- node [below] {\(\l\)} (\x{t}, \h)
                    -- (\p{t})
        ;
    }

    \draw [loosely dotted, thick, dash phase=20pt]
        let
            \p{s} = (end),
        in
            (\x{s}, \blockHeight)
                -- ($(\x{s}, \blockHeight) + (.5, 0)$)
            (\x{s}, -\blockHeight)
                -- ($(\x{s}, -\blockHeight) + (.5, 0)$)
    ;

    \draw[timeline]
        (0, 0) -- (end)
    ;
    \draw [loosely dotted, thick, dash phase=20pt]
        (end) -- ($(end) + (0.5, 0)$)
    ;
\end{tikzpicture}

%% file: figures/blocks/extended/pi.tex
\begin{tikzpicture}
    \coordinate (startB1)   at (0.5, 0);
    \coordinate (1toB1)     at (1.7, 0);
    \coordinate (startB2)   at (0.7, 0);
    \coordinate (1toB2)     at (2.7, 0);
    \coordinate (startB3)   at (1.7, 0);
    \coordinate (1toB3)     at (2.9, 0);
    \coordinate (2toB3)     at (3.2, 0);
    \coordinate (end)       at (3.2, 0);
    
    \foreach \s/\t/\h/\l in {startB2/1toB2/-0.5/B_2, startB3/1toB3/0.5/B_3} {
        \draw[block]
            let
                \p{s} = (\s),
                \p{t} = (\t)
            in
                (\p{s})
                    -- node [left=-0.1] {\(\l\)} (\x{s}, \h)
                    -- (\x{t}, \h)
                    -- (\p{t})
        ;
    }

    \foreach \s/\t/\h/\type/\l in {
        startB1/1toB1/1.2/nullTimerKilled/B_1
    } {
        \draw[block]
            let
                \p{s} = (\s),
            in
                (\p{s})
                    -- node [left=-0.1] {\(\l\)} (\x{s}, \h)
        ;
        \draw[blockBeforeKilled]
            let
                \p{s} = (\s),
                \p{t} = (\t),
            in
                (\x{s}, \h)
                    -- (\x{t}, \h)
                node[\type] at (\x{t}, \h) {}
        ;
    }

    \foreach \point/\lbl/\position/\h in {
        startB1/i_1/above/1.2,
        startB2/i_2/below/-0.5,
        startB3/i_3/above/0.5,
        1toB1/i_4/above/1.2,
        1toB2/i_5/below/-0.5,
        1toB3/i_6/above/0.5
    } {
        \draw
            let
                \p{s} = (\point)
            in
                node [\position] at (\x{s}, \h) {\(\lbl\)}
        ;
    }

    \draw[timeline]
        (0, 0) -- (startB1)
        (startB2) -- (1toB2)
        (1toB3) -- (end)
    ;
    \draw[hole]
        (startB1) -- (startB2)
        (1toB2) -- (1toB3)
    ;
\end{tikzpicture}

%% file: figures/blocks/pspace-hard/timerY.tex
\begin{tikzpicture}
    \coordinate (startBx)   at (1, 0);
    \coordinate (1toBx)     at (4, 0);
    \coordinate (2toBx)     at (7, 0);
    \coordinate (startBy1)  at (1.2, 0);
    \coordinate (1toBy1)    at (2.2, 0);
    \coordinate (startBy2)  at (2.8, 0);
    \coordinate (1toBy2)    at (3.8, 0);
    \coordinate (end)       at (8, 0);
    
    \foreach \s/\t/\h/\l in {
        startBx/1toBx/0.5/B_x,
        1toBx/2toBx/0.5/,
        startBy1/1toBy1/-0.5/B_y^1,
        startBy2/1toBy2/-0.5/B_y^2} {
        \draw[block]
            let
                \p{s} = (\s),
                \p{t} = (\t)
            in
                (\p{s})
                    -- node [left] {\(\l\)} (\x{s}, \h)
                    -- (\x{t}, \h)
                    -- (\p{t})
        ;
    }

    \draw[timeline]
        (0, 0) -- (startBx)
        (startBx) -- (startBy1)
        (startBy1) -- (1toBy2)
        (1toBy2) -- (1toBx)
        (1toBx) -- (2toBx)
        (2toBx) -- (end)
    ;
\end{tikzpicture}